\documentclass{llncs}

\newif\ifieee
\ieeefalse               

\newif\ifllncs
\llncstrue             

\newif\ifacmconf
\acmconffalse           

\newif\iflongversion
\newcommand{\vlong}[1]{\iflongversion #1\fi}
\newcommand{\vshort}[1]{\iflongversion \else #1\fi}

\usepackage{ifthen}   
\usepackage{cite}

\usepackage{amsmath}
\usepackage{amscd}
\usepackage{amsfonts} 
\usepackage{amssymb}
\usepackage{latexsym}  

\usepackage[mathcal]{euscript}
\usepackage{ifsym}

\usepackage{xspace}    
\usepackage{epsfig}
\usepackage[latin1]{inputenc} 
\usepackage{verbatim}  
\usepackage{stmaryrd}
\usepackage{xspace}
\usepackage{eepic}
\usepackage{epsfig}
\ifieee
\usepackage{prooftree}
\usepackage{synttree}
\usepackage{parsetree}
\fi

\usepackage{color}
\usepackage{graphicx}
\usepackage{pgf}
\usepackage{tikz}
\usetikzlibrary{arrows,automata,topaths}

\usepackage{hyperref} 
\hypersetup{%
  colorlinks=true,
}

\newcount\timehh\newcount\timemm \timehh=\time 
\divide\timehh by 60 \timemm=\time
\newcommand{\timestamp}{ 
  {\protect\small\sl\today\ -- 
    \ifnum\timehh<10 0\fi\number\timehh\,:\,
    \ifnum\timemm<10 0\fi\number\timemm}}

\newif\ifsolutions

\newcommand{\squishlist}{
   \begin{list}{$\bullet$}
    { \setlength{\itemsep}{0pt}      \setlength{\parsep}{3pt}
      \setlength{\topsep}{3pt}       \setlength{\partopsep}{0pt}
      \setlength{\leftmargin}{1.5em} \setlength{\labelwidth}{1em}
      \setlength{\labelsep}{0.5em} } }

\newcommand{\squishlisttwo}{
   \begin{list}{$\bullet$}
    { \setlength{\itemsep}{0pt}    \setlength{\parsep}{0pt}
      \setlength{\topsep}{0pt}     \setlength{\partopsep}{0pt}
      \setlength{\leftmargin}{2em} \setlength{\labelwidth}{1.5em}
      \setlength{\labelsep}{0.5em} } }

\newcommand{\squishend}{
    \end{list}  }

\long\def\ignore#1{\relax}

\def\refsec#1{Section~\ref{#1}}

\def\reflem#1{Lemma~\ref{#1}}

\def\refthm#1{Theorem~\ref{#1}}
\def\refcor#1{Corollary~\ref{#1}}

\usepackage{marginnote}

\newcommand{\zs}{\ensuremath{\vec{z}}\xspace}

\newcommand{\Vars} {\mbox{\it Vars}}
\newcommand{\vars} {\Vars}

\renewcommand{\int}{\ensuremath{\mathbb{Z}}\xspace}

\newcommand{\rat}{\ensuremath{\mathbb{Q}}\xspace}

\newcommand{\env}{\ensuremath{\eta}\xspace}

\newcommand{\proves}{\vdash}

\newcommand{\Rea}{\ensuremath{\mathbb{R}}\xspace}

\newcommand{\one}{\ensuremath{\mathbf{1}\xspace}}

\newcommand{\set}[1]{\ensuremath{ \{ #1 \} }} 
\newcommand{\dset}[2]{ \{ {#1} \mid {#2} \} } 
 
\newcommand{\seq}[1]{ \langle {#1}  \rangle }

\renewcommand{\phi}{\varphi}

\newcommand{\Ra}{\Rightarrow}
\newcommand{\ra}{\rightarrow}

\ifieee
\usepackage{amsthm}
\fi

\ifieee
\theoremstyle{plain} 
\newtheorem{theorem}{Theorem}[section]
\newtheorem*{theorem*}{Theorem}
\newtheorem{lemma}[theorem]{Lemma}
\newtheorem{lemma*}{Lemma}

\newtheorem*{proposition*}{Proposition}
\newtheorem{corollary}[theorem]{Corollary}
\newtheorem*{corollary*}{Corollary}

\theoremstyle{definition} 
\newtheorem{definition}[theorem]{Definition}
\newtheorem*{definition*}{Definition}

\newtheorem*{notation*}{Notation}

\newtheorem*{remark*}{Remark}

\newtheorem*{example*}{Example}

\newtheorem*{examples*}{Examples}

\newtheorem*{note*}{Note}

\newtheorem*{notes*}{Notes}

\newtheorem*{exercise*}{Exercise}

\newtheorem*{exercises*}{Exercises}

\fi

\ifllncs
\else\fi

\def\squareforqed{\hbox{{///}}}

\newcommand{\dhthy}{\ensuremath{\mathsf{AG}\hat{\ }}\xspace} 

\newcommand{\NZE}{\ensuremath{\mathop{NZE}}\xspace}

\newcommand{\nzmult}{\ensuremath{\mathbin{**}}\xspace}
\newcommand{\gtoe}[1]{[#1]}

\newcommand{\mdl}{\ensuremath{\mathcal{M}\xspace}}

\newcommand{\mdlF}{\ensuremath{\mdl_{F}}\xspace}
\newcommand{\ultraD}{\ensuremath{\mdl_{D}}\xspace}
\newcommand{\mdlQ}{\ensuremath{\mdl_{\rat}}\xspace}

\newcommand{\fieldq}{\ensuremath{\mathbb{F}_{q}}\xspace}

\newcommand{\fieldprod}{\ensuremath{\mathbb{F}_D}\xspace}

\newcommand{\boxfn}{\ensuremath{\gtoe{\cdot}}\xspace}

\newcommand{\Ind}{\ensuremath{\mathrm{Ind}}}
\newcommand{\Gen}[1]{\ensuremath{\mathrm{Gen}(#1)}}

\newcommand{\gop}{\cdot}  
\newcommand{\ginv}{\mathit{inv}}
\ifieee\newcommand{\gid}{\ensuremath{\mathit{id}\xspace}}\else\renewcommand{\gid}{\ensuremath{1}\xspace}\fi

\newcommand{\eid}{\ensuremath{1}\xspace}
\newcommand{\emult}{\;* \;}
\newcommand{\eone}{1}
\newcommand{\einv}{\mathit{i}}

\newcommand{\eadd}{+}
\newcommand{\ezero}{0}
\newcommand{\eneg}{-}

\usepackage{xy} 
\xyoption{all} 

\makeatletter
\renewcommand\paragraph{\@startsection{paragraph}{4}{\z@}%
                       {-8\p@ \@plus -4\p@ \@minus -4\p@}%
                       {-0.5em \@plus -0.22em \@minus -0.1em}%
                       {\normalfont\normalsize\bfseries\boldmath}}
\makeatother

\newcommand{\skel}{\ensuremath{\mathbb{A}}}

\newcommand{\cons}{\parallel} 
\newcommand{\enc}[2]{\{\!|#1|\!\}_{#2}}
\newcommand{\tagged}[2]{[\![\,#1\,]\!]_{\sk(#2)}}

\newcommand{\newop}[2]{\newcommand{#1}{\ensuremath{\mathop{\mathsf{#2}}}}}
\newop{\CA}{CA}
\newop{\simple}{simple}
\newop{\visin}{visIn}
\newop{\component}{compnt}

\newop{\node}{node}
\newop{\xmit}{xmit}
\newop{\recv}{recv}
\newop{\neutral}{neutral}

\newop{\adv}{adv}
\newop{\reg}{reg}
\newop{\msg}{msg}
\newop{\orig}{orig}
\newop{\unique}{unique}
\newop{\non}{non}

\newop{\akey}{akey}
\newop{\skey}{skey}
\newop{\pname}{pname}
\newop{\data}{data}
\newop{\pubkey}{pk}
\newop{\privkey}{privkey}
\newop{\sigkey}{sigkey}
\newop{\verkey}{verkey}

\newop{\advOrig}{advOrig}
\newop{\advEPT}{advEPT}
\newop{\advEK}{advEK}
\newop{\advECT}{advECT}

\newop{\advDCT}{advDCT}
\newop{\advDK}{advDK}
\newop{\advDPT}{advDPT}

\newop{\advOrigVal}{advOrigVal}
\newop{\advPTVal}{advPTVal}
\newop{\advKVal}{advKVal}
\newop{\advCTVal}{advCTVal}

\newop{\nsInitone}{nsInit1}
\newop{\nsInittwo}{nsInit2}
\newop{\nsInitthree}{nsInit3}

\newop{\nsRespone}{nsResp1}
\newop{\nsResptwo}{nsResp2}
\newop{\nsRespthree}{nsResp3}

\newop{\lsnone}{lsn1}
\newop{\hear}{hear}

\newop{\self}{self}
\newop{\peer}{peer}
\newop{\mynonce}{mynonce}
\newop{\yrnonce}{yrnonce}

\newop{\kdf}{kdf}
\newop{\sig}{sig}

\newop{\pk}{pk}
\newop{\sk}{sk}
\newop{\vum}{VUM}
\newop{\um}{UM}
\newop{\mqv}{MQV}
\newop{\kea}{KEA}
\newop{\cremers}{CF}

\newop{\umkey}{K_{um}}
\newop{\mqvkey}{K_{mqv}}

\newop{\nf}{nf}

\newcommand{\indZero}{\ensuremath{\mathbf{0}}}
\newcommand{\indX}{\ensuremath{\mathbf{1}_x}}
\newcommand{\indY}{\ensuremath{\mathbf{1}_y}}
\newcommand{\indA}{\ensuremath{\mathbf{1}_a}}
\newcommand{\indB}{\ensuremath{\mathbf{1}_b}}

\newcommand{\iadh}{\textsc{iadh}}

\newcommand{\ingredient}{\sqsubseteq}

\newcommand{\bnd}{\ensuremath{\mathcal{B}}}

\newcommand{\frameit}[1]{\textcolor{blue}{\framebox{\textcolor{black}{ #1 }}}}

\newcounter{axiomnumber}

\newtheorem{goal}[theorem]{Security Goal}
\newtheorem{assumption}[theorem]{Assumption}

\def\squareforqed{\hbox{\rlap{$\sqcap$}$\sqcup$}}
\def\qed{\ifmmode\squareforqed\else{\vskip-\lastskip\nobreak\hfil
\penalty50\hskip1em\null\nobreak\hfil\squareforqed
\parfillskip=0pt\finalhyphendemerits=0\endgraf}\fi}
\renewcommand\paragraph[1]{\par\bigskip\noindent\textbf{#1}\hspace{1ex}}

\longversiontrue       

\long\def \vs-pf-msg {
  \begin{proof} Please see the Appendix.
  \end{proof}
}

\ifacmconf
\fi

\ifieee

\fi

\ifllncs\institute{Worcester Polytechnic Institute \\ \{dd,guttman\}@wpi.edu}
\fi

\begin{document}

\title{Symbolic Protocol Analysis for Diffie-Hellman%
}

\author{Daniel J.~Dougherty\qquad Joshua D. Guttman\ifieee
\else\thanks{We gratefully acknowledge support by
    the National Science Foundation under grant CNS-0952287.  Version
    of \today.}\fi}

\pagestyle{plain}
\pagenumbering{arabic}

\maketitle
\begin{abstract}
  We extend symbolic protocol analysis to apply to protocols using
  Diffie-Hellman operations.  Diffie-Hellman operations act on a
  cyclic group of prime order, together with an exponentiation
  operator.  The exponents form a finite field.  This rich algebraic
  structure has resisted previous symbolic approaches.

  We work in an algebra defined by the normal forms of a rewriting
  theory (modulo associativity and commutativity).  These normal forms
  allow us to define our crucial notion of \emph{indicator}, a vector
  of integers that summarizes how many times each secret exponent
  appears in a message.  We prove that the adversary can never
  construct a message with a new indicator in our adversary model.  

  Using this invariant, we prove the main security goals achieved by
  several different protocols that use Diffie-Hellman operators in
  subtle ways.

  We also give a model-theoretic justification of our rewriting
  theory: the theory proves all equations that are uniformly true as
  the order of the cyclic group varies.

  %
  %
\end{abstract}

\iftrue{}\else
\newpage\tableofcontents\newpage
\renewcommand{\paragraph}{\subsection} 
\fi


\section{Introduction}
\label{sec:intro}

Despite vigorous research in symbolic methods for cryptographic
protocol analysis, many gaps and limitations remain.  While systems
such as NPA-Maude~\cite{escobar2009maude}, ProVerif~\cite{Blanchet01},
CPSA\cite{cpsa09}, and Scyther~\cite{Cremers06} are extremely useful,
great ingenuity is still needed---as for instance
in~\cite{kuesters2009using}---to analyze protocols that use
fundamental cryptographic ideas such as Diffie-Hellman key agreement
(henceforth, DH).  Moreover, important types of protocols, such as
implicitly authenticated key-agreement, appear to be out of reach of
known symbolic techniques.  Indeed, for these protocols, computational
techniques have also led to considerable controversy, with arduous
proofs that provide little
confidence~\cite{kaliski2001unknown,krawczyk2005hmqv,%
  kunz2006security,menezes2007another}.

In this paper we present foundational results and a new analysis
technique that together expand the range of applicability of symbolic
analysis.  In preparation for stating our contributions we remind the
reader of the basics of the Diffie-Hellman key
exchange~\cite{DiffieHellman76}.  In the protocol's original form, the
principals $A,B$ agree on a suitable prime $p$, and a generator $1<g<p$
such that the powers of $g$ form a cyclic group of some large prime
order $q$.  For a particular session, $A$ and $B$ choose random values
$x,y$ respectively, raising a base $g$ to these powers  mod $p$:
\begin{equation}
  \label{eq:dh:basic}
   \xymatrix{A,x & \bullet\ar[r]^{g^x} & \qquad & \bullet\ar[l]_{g^y}
     & B, y}
\end{equation}
They can then both compute the value $(g^y)^x=g^{xy}=(g^x)^y$ (modulo
$p$, as we will no longer explicitly repeat).  We can regard $g^{xy}$
as a new shared secret for $A,B$.  This is reasonable because of the
Decisional Diffie-Hellman assumption (DDH), which is the assumption
that $g^{xy}$ is indistinguishable from the $g^z$ we would get from a
randomly chosen $z$, for any observer who was given neither $x$ nor
$y$.  The protocol is thus secure against a passive adversary, who
observes what the compliant principals do, but can neither create
messages nor alter (or misdirect) messages of compliant principals.

However, an active adversary can choose its own $x',y'$, sending
$g^{y'}$ to $A$ instead of $g^{y}$, and sending $g^{x'}$ to $B$ instead
of $g^{x}$.  Now, each of $A,B$ actually shares one key with the
adversary, who can act as a man in the middle, re-encrypting messages in
any conversation between $A$ and $B$.  Various protocols have been
proposed to achieve a range of security goals in the presence of an
active attacker, such as \emph{implicit authentication}, \emph{forward
  secrecy}, and \emph{preventing impersonation attacks}.  In
Section~\ref{sec:some:protocols} we describe some of these protocols.


The algebra of the structures on which DH protocols operate has been an
obstacle to analyzing them  These structures are cyclic
groups of some prime order $q$, together with an exponentiation
operator.  The exponents are integers modulo the prime $q$, which form
a field of characteristic $q$.  We will call such structures
\emph{DH-structures}.  The algebraic richness of DH-structures has
resisted full symbolic formalization, despite substantial steps for
subalgebras~\cite{escobar2009maude,kapur2003unification,kuesters2009using}.


In this paper, we make five contributions.  
\begin{enumerate}
\item \label{contrib:goals} 
  We represent security goals as logical formulas about transmission and
  reception events, together with freshness and non-compromise
  assumptions. 
  These clean, structural definitions are easy to work with, in contrast
  with the procedural notations prevalent among cryptographers.  They
  are based on strand spaces~\cite{ThayerHerzogGuttman99,Guttman10} as a
  model of protocol execution.
  \item \label{contrib:ultra} %
  We give a new treatment of the values used for DH exchanges.  These
  values are characterized by a set of equations, namely the equations
  $s=t$ that are valid in infinitely many DH-structures.  In fact, we
  prove that if an equation holds in infinitely many DH-structures
  then it holds in all of them.
  \begin{enumerate}
  \item Using an ultraproduct construction, we build a \emph{single}
    model $\ultraD$ that realizes precisely those equations true in all
    (equivalently, infinitely many) DH-structures.
    (\refthm{full-tfae})
  \item We define an equational theory $\dhthy$ that can be presented by
    a rewrite system that is terminating and confluent modulo
    associativity and commutativity (\refthm{thm:sn+cr}). %
    Furthermore, for all equations $s=t$, $\dhthy$ rewrites $s$ and $t$
    to the same normal form if and only if $s=t$ is true in $\ultraD$
    for all values of its free variables (\refthm{full-tfae}). %
    The normal forms of this rewrite system represent the messages.
  \end{enumerate}
\item The theory of DH-structures suggests an adversary model
  (\refsec{sec:protocols}).  The \emph{uniformly algebraic adversary} is
  the Dolev-Yao adversary augmented with the functions in the signature
  of DH-structures.  These functions are governed by the equations $s=t$
  derivable in $\dhthy$.  Thus---given the correspondence between
  $\dhthy$ and truth in DH-structures---the adversary can rely on any
  equation that, as the size of the underlying cyclic group grows, is
  valid infinitely often.
  \item Using the $\dhthy$ normal forms, we define the
  \emph{indicators} of a message.  Indicators count occurrences of
  secret values in exponents.  We prove that the adversary cannot
  create a message with a new indicator.  If the adversary transmits a
  message with a particular indicator, then it must have received some
  message with that indicator previously
  (\refthm{thm:indicators:preserved}).
 This invariant extends the Honest Ideal
  theorem~\cite{ThayerHerzogGuttman99} to the algebra $\dhthy$.  It is
  our primary proof method.
  \item To illustrate the power of our method, we prove about a dozen
  different security goals for three protocols
  (Sections~\ref{sec:analysis} and~\ref{sec:implicit}).
 These implicitly
  authenticated DH protocols have previously resisted attempts to give concise,
  convincing proofs of the goals they achieve.
  We also use our method to show why certain protocols do not meet
  some goals, matching the relevant attacks from the cryptographic
  literature.
    
\end{enumerate}

The set of indicators of a message is a set of vectors that count how
many times uncompromised values appear in exponents.  They are a
refinement of the standard notion of an atom occurring in term, needed
since our terms are considered modulo equations.  For instance,
suppose that in some execution, the exponents $a,b,x,y$ are assumed
uncompromised, where $x,y$ are ephemeral secret values and $a,b$ are
long-term secret values.  The sequence $\seq{a,b,x,y}$ determines a
\emph{basis} for writing these indicator vectors.

Relative to this basis, the factor $g^{xy}$ has indicator
$\seq{0,0,1,1}$ because $a,b$ appear 0 times each, and $x,y$ appear
once each.  $g^{x/y}$ would have indicator $\seq{0,0,1,-1}$, since $y$
appears $-1$ times, i.e.~inverted.  The factor $g^{ax}$ has indicator
$\seq{1,0,1,0}$ since $a,x$ appear once.  When we multiply factors, we
take unions of indicators.  Thus, $g^{xy}g^{bx}g^{ay}g^{ab}$ has
indicators
\ifieee
\else{
$$\{\seq{0,0,1,1},\seq{0,1,1,0},\seq{1,0,0,1},\seq{1,1,0,0}\}.$$}\fi
There is good motivation for protocols in which each non-zero integer
in an indicator is $\pm 1$~\cite{BressonEtAl11}.

In our model, when the indicator basis consists of uncompromised
exponents, adversary actions never produce any message containing any
new indicator (Theorem~\ref{thm:indicators:preserved}).  If the
adversary \emph{transmits} a message with some indicator vector
$\vec{v}$, then it previously \emph{received} some message with that
indicator vector $\vec{v}$.  Only the regular, non-adversary,
participants can emit messages with new indicators.

This idea, which is natural and appealing for DH, is challenging to
justify, which is probably why it is not familiar from the
cryptographic literature.  Its soundness as a proof technique rests on
our foundational results concerning DH-structures
(contribution~\ref{contrib:ultra}).

\paragraph{Structure of this paper.}  We next, in
Sec.~\ref{sec:some:protocols}, introduce a few protocols we will use
as running examples.  Sec.~\ref{sec:strands} introduces the strand
space theory, and the Sec.~\ref{sec:an-equational-theory} presents our
equational theory $\dhthy$.  We use strand spaces in
Sec.~\ref{sec:protocols} to define {\iadh} protocols and the adversary
actions.  Section~\ref{sec:adversary} proves the key limitative
theorem on indicators and the adversary.  Sec.~\ref{sec:analysis}
defines a variety of security goals for {\iadh} protocols, and applies
the key limitative result to establish these goals; the focus shifts
specifically to implicit authentication in Sec.~\ref{sec:implicit}.
Sec.~\ref{sec:algebra} takes a model-theoretic point of view on
DH-structures and proves completeness of the theory \dhthy.  In
Sec.~\ref{sec:conc}, we comment on some related work and conclude.

\vshort{Proofs omitted from the main body are in the appendix for the
  referees' convenience.  A long version is available at
  \url{http://web.cs.wpi.edu/~guttman/pubs/iadh.pdf}.}


\section{Some Protocols of Interest}

\label{sec:some:protocols}

We start by describing some illustrative protocols at the level of
detail typically seen in the literature.  In order not to prejudice
ourselves in evaluating possible attacks, we will write $R_B$ for the
public value that $A$ receives, purportedly from $B$, rather than
writing $g^y$, since no one yet knows whether it is the same value
that $B$ sent.  We likewise write $R_A$ for the public value that $B$
receives, purportedly from $A$.  The participants hope that $R_A=g^x$
and $R_B=g^y$.

The Station-to-Station protocol~\cite{diffie1992authentication}
authenticates the Diffie-Hellman exchange by digital signatures on the
exchange.  In a simplified STS, the exchange in Eqn.~\ref{eq:dh:basic}
is followed by the signed messages:
\begin{equation}
  \label{eq:dh:sts:auth}
   \xymatrix{A & \bullet\ar[r]^{\tagged{g^x\cons R_B}A} & & &
     \bullet\ar[l]_{\tagged{g^y\cons R_A}B} & B} 
\end{equation}
The signatures\footnote{We write $t\cons t'$ for the concatenation of
  $t$ with $t'$.  A digitally signed message ${\tagged{t}A}$ means
  $t\cons\sig(H(t),\sk(A))$, where $\sig$ is a signature algorithm,
  $H(t)$ is a hash of $t$, and $\sk(A)$ is a signing key owned by
  $A$.}
exclude a man in the middle, assuming some public key infrastructure
to certify $\sk(A),\sk(B)$.
The costs of STS includes an additional message transmission and
reception for each participant, in each session.  Moreover, each
participant must also prepare one digital signature and also verify
one digital signature specifically for that session.  There is also a
privacy concern, since the signatures publicly associate $A$ and $B$
in a shared session.

An alternative to using per-session digital signatures is
\emph{implicit authentication}~\cite{blake1999authenticated}.  Here
the goal is to ensure that any principal that can compute the same
value as $A$ \emph{can only} be $B$, and conversely.  To implement
this idea, each principal maintains a long-term secret, which we will
write as $a$ for principal $A$, and as $b$ for $B$; they publish the long-term
public values $g^a,g^b$, which we will refer to as $Y_A,Y_B$, etc.
The trick is to build the use of the private values $a,b$ into the
computation of the shared secret, so that only $A,B$ can do it.  In
the ``Unified Model'' {\um} of Ankney, Johnson, and
Matyas~\cite{ankney1995unified}, the principals combine long term
values with short term values by concatenating and hashing.  They send
only the messages shown in Eqn.~\ref{eq:dh:basic}, and then---letting
$H(x)$ be a hash of $x$---compute their keys:
\begin{equation}
  \label{eq:dh:unif:model:key:comp:local}
  {A:\; k=H({Y_B}^a\cons {R_B}^x) \quad B:\; k=H({Y_A}^b\cons {R_A}^y)},
\end{equation}
obtaining the shared value $H(g^{ab}\cons g^{xy})$ if $R_A=g^x$ and
$R_B=g^y$.
%
%
Again, public key infrastructure must associate the public value $Y_P$
with the intended peer $P$.  However, no digital signature is
generated or checked specific to this run.  If $A$ has frequent
sessions with $B$, $A$ can amortize the cost of the certificate
verification by keeping $Y_B$ in secure storage.

Menezes-Qu-Vanstone ({\mqv})~\cite{law2003efficient} relies only on
algebraic operations.  {\mqv} computes the key via the rules:
\begin{equation}
  \label{eq:mqv:key:comp:local}
  {A:\; k=({R_B}\cdot{Y_B}^{\gtoe{{R_B}}})^{s_A} \quad 
    B:\; k= ({R_A}\cdot{Y_A}^{\gtoe{{R_A}}})^{s_B}} 
\end{equation}
where $s_A=x+a\gtoe{{R_A}}$ and $s_B=y+b\gtoe{{R_B}}$. The ``box''
operator coerces numbers mod $p$ to a convenient form in which they
can be used as exponents. In the literature this is written in the
typographically more cumbersome form of a bar, as $\overline{R_B}$. In
a successful run, $A$ obtains the value 
\begin{equation}
  \label{eq:mqv:key:comp:eqns}
  ({g^y}\cdot(g^b)^{\gtoe{g^y}})^{s_A}  = 
  (g^{(y+b{\gtoe{g^y}})})^{(x+a\gtoe{g^x})} = g^{(s_B\cdot
    s_A)} 
\end{equation}
and $B$ obtains $g^{s_A\cdot s_B}$, which is the same value.  {\mqv}
differs from {\um} only in the function that the principals use to
compute the key.  {\mqv}'s key computation makes it algebraically
challenging to model and to analyze.  Controversy
about its security remains~\cite{kaliski2001unknown,krawczyk2005hmqv,%
  kunz2006security,menezes2007another}.


\section{Background:  Strand Spaces}
\label{sec:strands} 

In this paper, we adopt the strand space formalism, although allowing
messages to form more complex algebraic structures than in earlier
papers, e.g.~\cite{ThayerHerzogGuttman99,Guttman10}.

\paragraph{Strands.}  A \emph{strand} is a sequence of local actions
called \emph{nodes}. A node may be either:
\begin{itemize}
\item a message \emph{transmission};
\item a message \emph{reception}; or else 
\item a \emph{neutral} node.  Neutral nodes are local events in
  which a principal consults or updates its local
  state~\cite{Guttman11a}.
\end{itemize}
If $n$ is a node, and the message $t$ is transmitted, received, or
coordinated with the state on $n$, then we write $t=\msg(n)$.
We write bullets $\bullet$ for transmission and reception events and
circles $\circ$ for neutral events, involving only the local state.
Double arrows indicate successive events on the same strand,
e.g. $\circ\Rightarrow\bullet\Rightarrow\bullet$.  

Each strand is either a \emph{regular strand}, which represents the
sequence of local actions made be a single principal in a single local
session of a protocol, or else an \emph{adversary strand}, which
represents a single action of the adversary.  
%

A \emph{protocol} is a set of regular strands, called the \emph{roles}
of the protocol.  We assume that every protocol contains a specific
role, called the \emph{listener} role, which consists of a single
reception node $n=\rightarrow \bullet$.  We use listener strands to
provide ``witnesses'' when $\msg(n)$ has been disclosed, especially to
specify confidentiality properties.

\emph{Adversary strands} consist of zero or more reception nodes
followed by one transmission node.  They represent the adversary
obtaining the transmitted value as a function of the values received;
or creating it, if there are no reception nodes.  All values that the
adversary handles are received or transmitted; none are silently
obtained from long-term state.  In fact, allowing the adversary to use
neutral nodes---or strands of other forms---provides no additional
power.  (See Section~\ref{sec:adversary}.)

We regard the messages transmitted and received on $\bullet$ nodes,
and obtained from long-term state on neutral nodes $\circ$, as forming
an abstract algebra.  \emph{Concatenation} and \emph{encryption} are
operators that construct values in the algebra from a pair of given
values, and we regard $v_0\cons v_1$ as equal to $u_0\cons u_1$ just
in case $v_0=u_0$ and $v_1=u_1$.  Similarly, $\enc{v_0}{v_1}$ equals
$\enc{u_0}{u_1}$ just in case $v_0=u_0$ and $v_1=u_1$.  That is, they
are \emph{free} operators.  For our present purposes, it suffices to
represent other operators such as hash functions and digital
signatures in terms of these.

The \emph{basic} values that are neither concatenations nor
encryptions include principal names; keys of various kinds; group
elements $x$, $x\cdot y$, and $g^x$; and text values.  We regard
variables (``indeterminates'') such as $x$ as values distinct from
values of other forms, e.g.~products $z\cdot y$, or from other
variables.  A variable represents a ``degree of freedom'' in a
description of some executions, which can be instantiated or
restricted.  It may also represent an independent choice, as $A$'s
choice of a group element $x$ to build $g^x$ is independent of $B$'s
choice of $y$.  DH algebras are defined later in this section as the
normal forms of an AC rewriting system.

\paragraph{Ingredients and origination.}  A value $t_1$ is an
\emph{ingredient} of another value $t_2$, written $t_1\ingredient
t_2$, if $t_1$ contributes to $t_2$ via concatenation or as the
plaintext of encryptions:  $\ingredient$ is the least reflexive,
transitive relation such that:
$$ t_1\ingredient t_1\cons t_2, \qquad t_2\ingredient t_1\cons t_2,
\qquad t_1\ingredient \enc{t_1}{t_2}.$$
By this definition, $t_2\ingredient \enc{t_1}{t_2}$ implies that
(anomalously) $t_2\ingredient t_1$.  
%
%
%
For basic (non-encrypted, non-concatenated) values $a,b$, we have
$a\ingredient b$ iff $a=b$.

A value $t$ \emph{originates} on a transmission node $n$ if
$t\ingredient\msg(n)$, so that it is an ingredient of the message sent
on $n$, but it was not an ingredient of any message earlier on the
same strand.  That is, $m\Rightarrow^{+} n$ implies
$t\not\ingredient\msg(m)$.  

A basic value is \emph{uniquely originating} in an execution if there
is exactly one node at which it originates.  Freshly chosen nonces or
DH values $g^x$ are typically assumed to be uniquely originating.  

A value is \emph{non-originating} if there is no node at which it
originates.  An uncompromised long term secret such as a signature key
or a private decryption key is assumed to be non-originating.  Because
adversary strands receive their arguments as incoming messages, an
adversary strand that encrypts a message receives its key as a
message, thus originating somewhere.  Decryption and signature
creation are similar.

The set of non-originating values is denoted \non; 
the set of  uniquely originating values is denoted \unique.

Very often in DH-style protocols unique origination and
non-origination are used in tandem.  When a compliant principal
generates a random $x$ and transmits $g^x$, the former will be
non-originating and the latter uniquely originating.

\paragraph{Executions are bundles.}  The strand space theory
formalizes protocol executions by \emph{bundles}.  A bundle is a
directed, acyclic graph.  Its vertices are nodes on some strands
(which may include both regular and adversary strands).  Its edges
include the strand succession edges $n_1\Ra n_2$, as well as
\emph{communication edges} written $n_1\ra n_2$.  Such a dag
$\bnd=(V,E_{\Ra}\cup E_{\ra})$ is a \emph{bundle} if it is causally
self-contained, meaning:
\begin{itemize}
\item If $n_2\in V$ and $n_1\Ra n_2$, then $n_1\in V$ and
  $(n_1,n_2)\in E_{\Ra}$;
\item If $n_2\in V$ is a reception node, then there is a unique
  transmission node $n_1\in V$ such that $\msg(n_2)=\msg(n_1)$ and
  $(n_1,n_2)\in E_{\ra}$;
\item The precedence ordering $\preceq_{\bnd}$ for $\bnd$, defined
  to be $(E_{\Ra}\cup E_{\ra})^*$, is a well-founded relation.
\end{itemize}
The first clause says that a node has a causal explanation from the
occurrence of the earlier nodes on its strand.  The second says that
any reception has the causal explanation that the message was obtained
from some particular transmission node.  The last clause says that
causality is globally well-founded.  It holds automatically in
\emph{finite} dags $\bnd$, which are the only ones we consider here.

When we assume that a value is non-originating, or uniquely
originating, we constrain which bundles $\bnd$ are of interest to us,
namely those in which the value originates on no node of $\bnd$, or on
one node of $\bnd$, respectively.


\section{An Equational Theory of Messages}
\label{sec:an-equational-theory}
%
%
%
%
As described in the Introduction, our challenge is to define an
equational theory that captures the relevant algebra of DH structures %
and admits a notion of reduction that supports modeling messages as
normal forms. 
By the Decisional Diffie-Hellman assumption, an adversary \emph{cannot}
retrieve the exponent $x$ from a value $g^x$ that a regular participant
has constructed.  This limitation is reflected in our formalism in a
straightforward way.  Namely, we do not provide a logarithm function in
the signature of DH-structures.

   In addition we must confront the fact that the
exponents in a DH
structure form a field, and fields cannot be axiomatized by equations.

Our strategy is as follows.  We work with a sort $G$ for base-group
elements and a sort $E$ for exponents.  The novelty is that we enrich
$E$ by adding a subsort $NZE$ whose intended interpretation is the non-0
elements of $E$. %

The device of expressing ``non-zero'' as a sort fits well with the
philosophy of capturing uniform capabilities algebraically.  For
instance no term which is a sum $e_1 + e_2$ will inhabit the sort $NZE$
because each finite field has finite characteristic and so there may be
instantiations of the variables in $e_1 + e_2$ driving the term to 0. On
the other hand, we will want to ensure that $NZE$ is closed under
multiplication; this is the role of the operator \nzmult below.

We show in this section that \dhthy admits a confluent and terminating
notion of reduction.  In section~\ref{sec:algebra} we prove a theorem that
describes the sense in which $\dhthy$ captures the equalities that hold
in almost all finite prime fields.



\begin{definition}
  \label{def:dhthy}
  The theory \dhthy is the equational theory comprising %
  the sorts and operation given in Table~\ref{signature} and %
  the equations given in Table~\ref{eqthy}.   
  We write $box(t)$ as $\gtoe{t}$, and we write $exp(t,e)$ and $t^e$.   
%

\end{definition} 
%

\begin{table}[tb]
  \centering
  \framebox{
    \begin{minipage}[c!]{.95\linewidth}   
      \centering {Sorts: $G$, $E$, and $NZE$, with $NZE$ a subsort of $E$;} 
      \begin{minipage}{.4\linewidth}
        \begin{align*}
         \gop & : G \times G \to G \\
          \gid & : \to G \\
          \ginv & : G \to G  \\
          exp: & : G \times E \to E \\
          box & : G \to NZE 
        \end{align*}
      \end{minipage}
      \begin{minipage}{.4\linewidth}
        \begin{align*}
         \eadd, \; \emult & : E \times E \to E \\
          \eone & : \to NZE  \\
          \einv & : NZE \to NZE \\
\\
          \nzmult & : NZE \to NZE
        \end{align*}
      \end{minipage}
    \end{minipage}
  } \\[2mm] 
  \caption{The signature for ${\dhthy}$}
  \label{signature}
\end{table}

\begin{table}[tb]
  \centering
  \fbox{
    \begin{minipage}[c!]{.95\linewidth}   
      \begin{minipage}{.4\linewidth}
        \centering {$(G, \gop, \ginv, \gid)$ \\ is an abelian group}
        \begin{eqnarray*}
          (a \gop b) \gop c &=& a \gop (b \gop c) \\
          a \gop b &=& b \gop a \\
          b \gop \gid &=& b \\ 
          b \gop \ginv(b) &=& \gid               
        \end{eqnarray*}
      \end{minipage}
      \hfill 
      \begin{minipage}[t!]{.55\linewidth}
        \centering{ $(E, \eadd, \ezero, \eneg, \emult, \one, \einv)$
          \\ is a commutative unitary ring}
        \begin{eqnarray*}
          (x \eadd y) \eadd z &=& x \eadd (y \eadd z) \\
          x \eadd y &=& y \eadd x \\
          x \eadd 0 &=& x \\
          x \eadd (\eneg x) &=& 0\\
          (x \emult y) \emult z &=& x \emult (y \emult z) \\
          x \emult y &=& y \emult x \\
          x \emult (y \eadd z) &=& (x \emult y) \eadd (z \eadd z) \\
          x \emult \eone &=& x 
        \end{eqnarray*}
      \end{minipage}  \\[3mm]
      \begin{minipage}[t!]{.5\linewidth}
        \centering {Multiplicative inverse, closure \\ at sort NZE}
        \begin{eqnarray*}
          u \nzmult v &=& u \emult v \\
          i (u * v) &=& i(u) * i(v) \\
          i ( 1 ) &=& 1 \\
          i ( i ( w ) ) &=& w 
        \end{eqnarray*}
      \end{minipage} \hfill 
      \begin{minipage}[t!]{.5\linewidth}
        \centering{Exponentiation makes $G$ \\ a unitary right
          $E$-module} 
        \begin{eqnarray*}
          (a ^ x)^y &=& a ^{x \emult y} \\ 
          a ^ {1} &=& a \\
          (a \gop b) ^ x &=& a^x \gop b^x \\
          a ^ {(x \eadd y)} &=& a ^ x \gop a ^ y \\
          \gid ^ x &=& \gid             
        \end{eqnarray*}
      \end{minipage}
    \end{minipage}
  } \\[2mm]
  \caption{The theory ${\dhthy}$}
  \label{eqthy}
\end{table}
We next construct an associative-commutative rewrite system from
${\dhthy}$.  We orient each equation in Table~\ref{eqthy} in the
left-to-write direction, except for the associativity and
commutativity of $\gop, \eadd,$ and $\emult$.  Confluence requires the
new rules shown in Table~\ref{rrrules}, corresponding to equations
derivable from \dhthy that are needed to join critical pairs.

\begin{definition}
  Let $R$ be the set of rewrite rules given in
  Table~\ref{eqthy}---read from left to right, but without
  associativity and commutativity---and in Table~\ref{rrrules}.  The
  rewrite relation $\to_{\dhthy}$ is rewriting with $R$ modulo
  associativity and commutativity of $\gop, \eadd$, and $\emult$.
\end{definition}

\begin{table}[tb]
  \centering
  \fbox{
    \begin{minipage}[t]{0.475\linewidth}
      \begin{align*}
        \text{ At sort G } \\
        inv ( \gid ) \; &\to \; \gid \\ 
        inv ( a \gop b ) \; &\to \; inv ( a ) \gop inv ( b ) \\ 
        inv ( inv ( b ) ) \; &\to \; b \\ 
        ( inv ( a )  ) ^ x \; &\to \; inv ( a ^ x) \\
        a ^ 0 \; &\to \; \gid \\
        a ^ {-(x)}  \; &\to \; inv (a ^ x) \\
      \end{align*}
    \end{minipage}
    \begin{minipage}[t]{0.475\linewidth}
      \begin{align*}
        \text { At sort E } \\
        - ( 0 ) \; &\to \; 0 \\ 
        - ( x + y ) \; &\to \; - ( x ) + (- ( y ) ) \\
        - ( - ( x ) ) \; &\to \; x \\ 
        0 * x \; &\to \; 0 \\
        - (x) * y \; &\to \; - (x * y) \\
      \end{align*}
    \end{minipage}
  }\\[2mm] 
  \caption{Additional rewrite rules for $\to_{\dhthy}$}
  \label{rrrules}
\end{table}

\begin{theorem} \label{thm:sn+cr}
  The reduction $\to_{\dhthy}$ is terminating and confluent modulo AC.
\end{theorem}
\vshort{ \vs-pf-msg }
\long\def \vlsn+cr {
  \begin{proof}
    Termination can be established using the AC-recursive path order
    defined by Rubio\cite{rubio_fully-syntactic} with a precedence  in
    which exponentiation is greater than inverse, which is in turn
    greater than multiplication (and 1).
    This has been verified with the Aprove termination tool \cite{aprove}.

    Then confluence follows from local confluence, which is established
    via a verification that all critical pairs are joinable.  This
    result has been confirmed with the Maude Church-Rosser
    Checker\cite{duran_church-rosser-checker}.
  \end{proof}
}
\vlong{ \vlsn+cr }

Terms that are irreducible with respect to $\to_{\dhthy}$ are called
\emph{normal forms.}
The following taxonomy of the normal forms will be crucial in what
follows, most of all in the definition of indicators, %
Definition~\ref{def:indicator}.
  The proof is a routine simultaneous induction over the size of $e$ and $t$.

\begin{lemma} 
  \begin{enumerate}
  \item 
  If $e:E$ is a normal form then $e$ is a sum 
 \mbox{$m_1 \eadd \dots \eadd m_n$}
 where \\
  (i) each $m_i$ is of the form
  $
  e_1 \emult \dots \emult e_k \quad k \geq 0
  $, %
(ii) no $e_i$ is of the form $\einv(e_j)$, and %
(iii) each $e_i$ is one of: 
\[ x, \quad \einv(x), \quad \gtoe{t}, \quad \einv(\gtoe{t})
\]
 with $x$ a $G$-variable and $t:G$  a $G$-normal form.

 The case $n=0$ is taken to mean $e = 0$; the case $k=0$ is taken to
 mean $m_i = \eid$
  We call terms of the form $m_i$  \emph{irreducible monomials}    

\item If $t: G$ is a normal form then $t$ is a product %
  \mbox{$ t_1 \gop \dots \gop t_n , \quad n \geq 0 $} %
  where \\
 (i) no $t_i$ is of the form $\ginv(t_j)$, and 
 (ii)  each $t_i$ is one of: 
    \[
    v \qquad \ginv(v) \qquad v^e \qquad \ginv(v^e)
    \]
    with $v$ a  $G$-variable $e:E$ an irreducible monomial.

The case $n=0$ is taken to mean $t = \gid$.
\end{enumerate}
\end{lemma}

\section{Formalizing the Protocols and the Adversary}
\label{sec:protocols}

We consider a collection of protocols that all involve the same
strands, i.e.~sequences of transmissions, receptions, and neutral
events.  They differ almost exclusively in the key computations used
to generate the shared secret.

{\mqv} and {\um}~\cite{blake1999authenticated} both fit our pattern.
Various other protocols fit this pattern with some cajoling.
{\kea}~\cite{blake1999authenticated} fits the pattern too, although
its key computation uses addition mod $p$ to combine $g^{ay}$ and
$g^{bx}$.  Cremers-Feltz's protocol
$\cremers$~\cite{CremersFeltz2011}, in which the shared secret is
$g^{(x+a)(y+b)}$, almost fits:  They use the digitally signed messages
$\tagged{R_A}A$ and $\tagged{R_B}B$.  Our analysis is equally
applicable in this case.

In these protocol descriptions, we make explicit aspects that are
normally left implicit.  One is the interaction with the certifying
authority.  Kaliski~\cite{kaliski2001unknown} argues that the
certification protocol should be considered in analysis, because the
correctness of forms of a protocol may depend on exactly what checks a
\textsc{ca} has actually made.  We will also show how the local
session interacts with the local principal state.  

\paragraph{The IADH initiator and responder roles.}  We summarize the
activities of regular initiators and responders in
Figure~\ref{fig:iadh}.
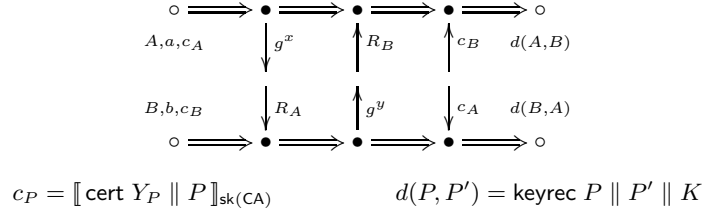
\begin{figure}[tb]
  \centering
  $$\xymatrix@R=6mm{
    \circ\ar@{=>}[r]\ar@{}[d]|{A,a,c_A} &
    \bullet\ar@{=>}[r]\ar[d]^{g^x} & \bullet\ar@{=>}[r] &
    \bullet\ar@{=>}[r] &
    \circ\ar@{}[d]|{d({A,B})}  \\
    &\null\ar[d]^{R_A} & \null\ar[u]_{R_B}&\null\ar[u]_{c_B}\ar[d]^{c_A} &\\
    \circ\ar@{=>}[r]\ar@{}[u]|{B,b,c_B} & \bullet\ar@{=>}[r] &
    \bullet\ar@{=>}[r]\ar[u]_{g^y} & \bullet\ar@{=>}[r] & \circ
    \ar@{}[u]|{d({B,A})} }$$
  $c_P=\tagged{\mathsf{cert}\; Y_P\cons P}{\CA}$ \hfil $d({P,P'})=
  {\mathsf{keyrec}\; P\cons P'\cons K}$
  \caption{{\iadh} Initiator and Responder Strands}
  \label{fig:iadh}
\end{figure}
We specify, for the initiator $A$:
\begin{enumerate}
  \item $A$ retrieves its principal name $A$, its long term secret
  $a$, and its public certificate $c_A$ from its secure storage.
%
  \item $A$ chooses a fresh ephemeral $x$, transmitting
  $R_A=g^x$.\label{clause:ephemeral:out}
  \item $A$ receives some $R_B$, which it checks to be a non-trivial
  group element, i.e.~a value of the form $g^y$ for some $y\not=0,1
  \bmod q$.\label{clause:ephemeral:in}
  \item It receives a certificate $c_B$ associating $Y_B$ with $B$'s
  identity.  We do not specify here how the participant determines
  what name $B$ to require in this certificate, or how it determines
  which \textsc{ca}s to accept.  This is implementation-dependent.
  \item Finally, $A$ performs the protocol-specific key computation to
  determine $K$.  $A$ checks the exponentiations yield non-1 values,
  and fails if any do.  On success, $A$ deposits a \emph{key record}
  into its local state database, so that $K$ may be used for a secure
  conversation between $A$ and $B$.
%
\end{enumerate}
In clause~\ref{clause:ephemeral:out}, $A$ chooses $x$ freshly.
Because $A$ never sends $x$ as an ingredient in any message---but only
$g^x$---it follows that $x$ has negligible probability of occurring in
a message.  After all, $A$ does not send it; any other regular
participant is overwhelmingly unlikely to choose the same value again;
and the adversary is overwhelmingly unlikely to choose it, e.g.~as a
guess.  For this reason we model $x$ as being ``non-originating''; for
the same reason, $g^x$ is declared to be uniquely originating.

We always add the assumptions that $x$ is non-originating and $g^x$ is
uniquely originating whenever a \emph{regular} strand selects
$R_A=g^x$.  In particular, since $x$ is a fresh, unconstrained choice
that the principal makes, we always instantiate it with a simple
value, essentially a \emph{parameter}, and never with a compound
expression like $y\cdot 1/z$.  Essentially, $x$ is a generator of the
algebra of normal forms of $\dhthy$.

A responder $B$ behaves in a corresponding fashion, with predictable
changes to the names of its parameters.  The only real change is that
it receives an ephemeral public value $R_A$ in step 2 before
generating its ephemeral secret $y$ and transmitting its ephemeral
public value $g^y$ in step 3.  We will assume that $y$ is
non-originating and $g^y$ is uniquely originating whenever a regular
responder strand selects $R_B=g^y$.

The parameters to an initiator strand are $A,B,a,x,Y_B,R_B$.  We write
them in this order, and refer (e.g.)~to the fourth parameter as $x$,
despite the fact that in different instances of the role have
different choices for the parameter $x$.  The parameters to a
responder strand are $A,B,b,y,Y_A,R_A$; thus, we will write the
(purported) initiator's name first, and the (actual, known)
responder's name second.

We make an assumption on the principal states, namely that the node
$\circ\; {A,a,c_A}$ starting an initiator or responder strand is
possible only if the same principal $A$ has on some earlier occasion
received a certificate $c_A$, and deposited it into its state.
Certificates do not emerge \emph{ex nihilo}.  Gathering our
assumptions on regular initiator and responder strands:
\begin{assumption}\label{assumption:regular:init:resp}
    Suppose that $\bnd$ is a bundle.
  \begin{enumerate}
    \item If $\bnd$ contains an initiator strand $s$ with parameters
    $A,B,a,x,Y_B,R_B$, then:
    \begin{enumerate}
      \item $x$ is non-originating, and $g^x$ is uniquely
      originating.\label{assumption:freshness:init} 
      \item $x$ is a parameter, not a compound expression.
      \item For some transmission node $n\in\bnd$,
      $c_A\ingredient\msg(n)$ and $n\preceq_{\bnd}s_1$, %
      where $s_1$ is the first node on strand $s$.
    \end{enumerate}
    \item Symmetrically for responder strands $s$ in $\bnd$, with
    parameters $A,B,b,y,Y_A,R_A$:
    \begin{enumerate}
      \item $y$ is non-originating, and $g^y$ is uniquely
      originating.\label{assumption:freshness:resp} 
      \item $y$ is a parameter, not a compound expression.
      \item For some transmission node $n\in\bnd$,
      $c_B\ingredient\msg(n)$ and $n\preceq_{\bnd}s_1$, where $s_1$ is
      the first node on strand $s$.
    \end{enumerate}
  \end{enumerate}
\end{assumption}
Our results do not depend on the specific ordering of events in
initiator and responder strands.  As long as the neutral node
retrieving the long term secret and certificate occurs before any of
the other events, and as long as the neutral node depositing the $K$
into the state occurs only after the other event, then our results
remain correct.  They also do not distinguish between initiator and
responder strands:  We would allow two initiator strands to succeed in
implicit authentication, for instance.

\paragraph{Key computation functions.}  The shared secret $K$ is
generated using different functions in different {\iadh} protocols.
In the Unified Model {\um}, the key is generated by
\begin{equation}
  \label{eq:dh:unif:model:key:comp}
  {A:\; k=H({Y_B}^a\cons {R_B}^x) \quad B:\; k=H({Y_A}^b\cons {R_A}^y)}
\end{equation}
In the optimistic case that $R_A=g^x$ and $R_B=g^y$
\begin{equation}
  \label{eq:dh:unif:model:key:comp:opt}
  K = H(g^{ab}\cons g^{xy})
\end{equation}
In {\mqv} the key is generated by 
 \begin{equation}
  \label{eq:mqv:key:comp}
  {A:\; K=({R_B}\cdot{Y_B}^{\gtoe{{R_B}}})^{s_A} \quad 
    B:\; K= ({R_A}\cdot{Y_A}^{\gtoe{{R_A}}})^{s_B}} 
\end{equation} 
(where ${s_A}={(x+a\gtoe{g^x})}$ and ${s_B}={(y+b\gtoe{g^y})}$), so
when $R_A=g^x$ and $R_B=g^y$ the principals compute
\begin{equation}
  \label{eq:dh:mqv:key:comp:opt}
  (g^{(y+b{\gtoe{g^y}})})^{(x+a\gtoe{g^x})} =  g^{xy} \cdot
  g^{xb\gtoe{g^y}} \cdot g^{ya\gtoe{g^x}} \cdot
  g^{ab\gtoe{g^x}\gtoe{g^y}} 
\end{equation}
The key computation for Cremers-Feltz $\cremers$---somewhat simplified
to make it more parallel to the $\um$ and $\mqv$ computations---is:
\begin{equation}
  \label{eq:dh:cremers:feltz:key:comp}
  {A:\; K=({R_B}\cdot{Y_B})^{(a+x)} \quad 
    B:\; K= ({R_A}\cdot{Y_A})^{(b+y)}}
\end{equation}
so that, in the same optimistic case, 
\begin{equation}
  \label{eq:dh:cremers:feltz:key:comp:opt}
  (g^y\cdot g^b)^{(x+a)} = g^{(y+b)(x+a)} = 
  g^{xy}\cdot g^{xb}\cdot g^{ya}\cdot g^{ab} 
\end{equation}

The occurrences of $a,b,x,y$ in these terms show us a contrast between
{\um} and the other two.  In the latter, all four pairs consisting of
one parameter from $a,x$ and one from $b,y$, appearing together, may
be found in the exponent of some factor of the final shared secret.
However, in {\um}, only two of these pairs appears.  This suggests
that {\um} is more fragile than the latter two, and this explains why
it is vulnerable to key compromise impersonation while the others are
not.

In Section~\ref{sec:algebra}, we will develop an algebraic theory to
justify this kind of analysis.
%

\paragraph{Requesting and issuing certificates.}  We also identify
protocol roles for requesting certificates from certificate
authorities, and for the CAs to issue them (Fig.~\ref{fig:cert}).  The
client makes a request with its name $P$ and public value $Y_P$, and,
if successful, receives a certificate which it deposits into its local state.
In its request, a compliant principal named $P$ always chooses a fresh
long-term secret $a$, and computes $Y=g^a$.  The CA, on receiving a
request, issues a certificate
\begin{figure}[tb]
  \centering
  $$\xymatrix@R=2mm@C=3mm{
    \bullet\ar@{=>}[rr]^{\mathsf{pop}}\ar[d] && \bullet\ar@{=>}[r] & \circ \\
    {\;\tagged{\mathsf{cert\_req}\; P\cons Y}P}\ar[d] &&
    {\tagged{\mathsf{cert}\; Y\cons P}{\CA}}\ar[u] \\
    \bullet\ar@{=>}[rr]_{\mathsf{pop}} &&\bullet\ar[u] }$$
  \caption{Strands for Certificate Requests}
  \label{fig:cert}
\end{figure}
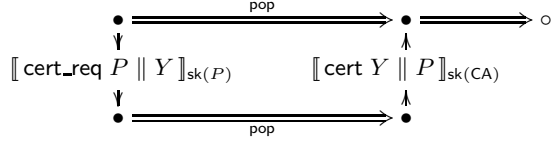
after a ``proof of possession'' protocol $\mathsf{pop}$ intended to
show $P$ possesses an $a$ such that $g^a=Y$.  We will not make
$\mathsf{pop}$ explicit.

We assume, whenever a bundle contains a regular certificate request,
that its $g^a=Y$ is uniquely originating.  Any subsequent use of $Y$
must obtain it through some sequence of message transmissions and
receptions tracing back, ultimately, to this originating node.  We
will not, however, always assume $a$ is non-originating, since
carelessness or malice may eventually lead to the disclosure of $a$.
Instead, if a particular $a$ is non-originating, we will explicitly
state that as a hypothesis in the security goals that depend on it.

We assume the {\CA} is uncompromised, i.e.~$\sk(\CA)\in\non$.  {\CA},
when receiving $Y$, should ensure that $Y\not=g^0,g^1$, and that it is
a member of the group (e.g.~via the little Fermat test).  Hence, there
is an $e$ such that $Y=g^e$.

Moreover, a successful $\mathsf{pop}$ means that the requester
possesses an exponent $e$ such that $Y=g^e$.  The requester is either
a regular participant or the adversary.  Thus, either:
\begin{itemize}
  \item $e$ is some parameter $a$, and $g^a$ originates uniquely on a
  certificate request strand; or else
  \item the request comes from an adversary strand, and $e$ is
  available to the adversary.
\end{itemize}
We will model the latter by assuming that the bundle containing this
certification generation strand also contains a listener strand
$n=\rightarrow \bullet$ with $\msg(n)=e$.  
%
%
%
%
\begin{assumption} 
  \label{assumption:regular:cert}  
  Let $\bnd$ be a bundle containing $\tagged{\mathsf{cert}\; Y\cons
    P}{\CA}$.  Assume $\sk(\CA)\in\non_{\bnd}$, and moreover:
  \begin{enumerate}
    \item For a certificate request strand, with parameters $P,a,\CA$:
    \begin{enumerate}
      \item $g^a$ originates uniquely;  
      \item $a$ is a parameter, not a compound expression.
    \end{enumerate}
    \item For a ${\CA}$ strand, with parameters $P,Y,\CA$:
    \begin{enumerate}
      \item There exists an $e\not=0,1$ such that
      $Y=g^e$;\label{assumption:regular:cert:in:group} 
      \item Either $Y=g^a$ where $g^a\in\unique_{\bnd}$, and $g^a$
      originates on a regular certificate request strand, or else
      there exists $n\in\bnd$ with
      $\msg(n)=e$.\label{assumption:regular:cert:unique:P}
%
%
%
%
    \end{enumerate}
  \end{enumerate}
\end{assumption}
By Clause~\ref{assumption:regular:cert:unique:P}, if
$a\in\non_{\bnd}$, then for at most one $P$ can a certificate
$\tagged{\mathsf{cert}\; g^a\cons P}{\CA}$ be issued.
The {\iadh} protocols are defined by the four roles shown in
Figs.~\ref{fig:iadh}--\ref{fig:cert}, using a key computation such as
those in {\mqv}, {\um}, and {\cremers}.

\paragraph{The Adversary.}
An adversary strand has zero or more reception nodes followed by a
transmission node:
\begin{definition}\label{def:adv}
  \emph{Adversary strands} take the forms:
    \begin{itemize}
      \item %
      Emission of a basic value  $a$:  $\langle +a
      \rangle $ %
      \item %
      Constructor strands:  $ \langle -a_1 \Ra \dots \Ra - a_n \Ra +t
      \rangle$ where is $t$ is in $\Gen{a_1, \dots, a_n}$
      \item %
      Destructor strands: $\langle -t \Ra +s_1 \dots \Ra +s_n \rangle $
      where $t$ is a concatenation of the values $s_i$.
      \item %
      Encryption strands: %
      $ \langle -K \Ra -t \Ra + \enc{t}{K} \rangle$
      \item %
      Decryption strands: %
      $ \langle -K^{-1} \Ra - \enc{t}{K} \Ra + t \rangle$
   \end{itemize}
   Suppose that $S_1, \dots, S_k$ are node-disjoint adversary strands.
   An \emph{adversary web}~\cite{Guttman10} using $S_1, \dots, S_k$ is
   an acyclic graph whose vertices are the nodes of the $S_i$, where
   for each edge $(n,n')$, either (i) $n \Ra n'$ on some strand or
   (ii) $n$ is a transmission node, $n'$ is a reception node, and
   $\msg(n) = \msg(n')$.
\end{definition}
%
%
This adversary model motivates a game between the adversary and the
system:
\begin{enumerate}
  \item The system chooses a security goal $\Phi$, involving secrecy,
  authentication, key compromise, etc., as in
  Figs.~\ref{fig:key:secrecy}--\ref{fig:impl:auth}.

  \item The adversary chooses a potential counterexample $\skel$
  consisting of regular strands with equations between values on the
  nodes, e.g.~an equation between a session key computed by one
  participant and a session key computed by another
  participant.\label{clause:adv:skel} 
  \item To show that $\skel$ can occur, the adversary chooses how to
  generate the messages in $\skel$.

  For each message reception node in $\skel$, the adversary must
  provide an acceptable message in time for that event.  The adversary
  benefits from transmission events on regular strands, which he can
  use to build messages for subsequent reception events.  For each
  reception node, the adversary chooses a recipe, consisting of an
  adversary web, using the strands of
  Def.~\ref{def:adv}.\label{clause:adv:strategy}

  This map---which, to every message reception event, associates an
  adversary web---is the \emph{adversary strategy}.

  The adversary strategy determines a set of equalities between a
  value computed by the adversary and a value $t$ ``expected'' by the
  recipient, or acceptable to the recipient.  They are the
  \emph{adversary's proposed equations}.
  \item The adversary wins if his proposed equations are valid in
  $(G_q,F_q)$, for infinitely many primes $q$.
\end{enumerate}
This game may seem too challenging for the adversary.  First, it wins
only if the equations are valid, i.e.~true for all instances of the
variables.  However, the adversary's proposed equations determine
polynomials, and each of these polynomials has a syntactically
determined degree $d$.  If it is not valid, it can have at most $d$
solutions, independent of the choice of $(G_q,F_q)$.  Hence, the set
of values for which the adversary's strategy works remains small,
regardless of how the cardinality of the structure $(G_q,F_q)$ grows.

Second, the adversary must choose how to generate all the messages,
its adversary strategy, before seeing any concrete bitstrings, or
indeed learning the prime $q$.  This objection motivates future
research into the \emph{computational soundness} of our approach.  The
hardness of DDH seems to suggest that the adversary acquires no useful
advantage from seeing the values $g^x$ etc.  Any definite claim would
require a reduction argument.



\section{Indicators}
\label{sec:adversary}

We turn now to a formal definition of indicators and the proof of a
key invariant that all adversary actions preserve.

Let $\int^k$ denote the set of all $k$-tuples of integers.  For
intuition about the following definition, think of $N$ as being a set
of \emph{non-originating values} for a bundle.  If $m$ is a monomial
occurring as a subterm of a term $t$, say that $m$ is
``maximal-monomial'' if $t$ has a subterm of the form $b^m$. 
%

\begin{definition}[Indicators]
  \label{def:indicator} 
  Let $N = \langle v_1, \dots, v_d \rangle$ be a vector of
  $\NZE$-variables.  If $m$ is an irreducible monomial, the
  \emph{$N$-vector} for $m$ is $\langle z_1, \dots, z_k \rangle$ where
  $z_i$ is the multiplicity of $v_i$ in $m$, counting occurrences of
  $i(v_i)$ negatively.

  If $e=m_1+\ldots +m_k$ is a term of type $E$, then $e$ is
  \emph{$N$-free} if each $m_i$ has $N$-vector $\seq{0,\ldots, 0}$.  

  When $t_0$ is any base term in normal form, then $\Ind_N(t_0)$ is
  the set of all vectors $\vec{z}$ such that $\vec{z}$ is the
  $N$-vector of $m$, where $m$ is a maximal-monomial subterm of
  $t_0$.  
    
  If $t=t_1\cons t_2$, then $\Ind_N(t)=\Ind_N(t_1)\cup\Ind_N(t_2)$.  

  If $t=\enc{t_1}{t_2}$, then $\Ind_N(t)=\Ind_N(t_1)$.  
\end{definition}
Thus, $\Ind_N(t)$ for a compound term $t$ is the union $\bigcup
\Ind_N(t_0)$, taking the union over all the base terms $t_0$ that are
ingredients of $t$, i.e.~$t_0\ingredient t$.

\smallskip\noindent\emph{Example:}  For $N = \seq{x, y}$, if $t$ is 
\[
  g^{x \; i(y)} \cdot g^{zx[g^x]} \cdot g^{xx[g^y]} ,
\]
then $\Ind_N(t) = \{ \langle 1,-1 \rangle, \langle 1,0 \rangle,
\langle 2,0 \rangle \}$.  The boxed values do not contribute to the
indicators.  

Since we often encounter indicators with no non-zero entries, we will
write $\indZero$ for this indicator $\seq{0,\ldots,0}$.  We will also
write $\indX$, $\indA$, etc., for the indicator that has a single $1$
in the position for that parameter, e.g.~for $\seq{0,0,1,0}$ and
$\seq{1,0,0,0}$ if the parameters are $a,b,x,y$ in that order.  Every
message sent in {\iadh} protocols is of this form:  All the indicator
weight is concentrated in at most a single 1.  A message $g^c$ with
$c\not\in\non$ has indicator $\indZero$.

Since the union $\bigcup \Ind_N(t_0)$ is over all the
\emph{ingredients} $t_0\ingredient t$, it does not include values used
only as keys in encryptions.  Thus, a protocol may compute a secret
such as $g^{xy}$ with an indicator $\seq{0,0,1,1}=\indX+\indY$, and
then applies a key derivation function, obtaining $k=\kdf(g^{xy})$.
If participants then send encrypted messages $\enc{t_1}k$, then it has
not transmitted a message with indicator $\indX+\indY$.

\begin{definition}
  Let $T = \set{t_1, \dots, t_k}$ be a set of terms.  The set
  $\Gen{T}$ \emph{generated} by $T$ is the least set of terms
  including $T$ and closed under the term-forming operations.
\end{definition}
\noindent{}The term-forming operations cannot cancel to reveal a
$v_i\in N$: 

\begin{theorem} \label{indicator} Suppose $T$ is a collection of
    terms such that every $e \in T$ of sort $E$ is $N$-free.  Then 
    \begin{enumerate}
       \item 
        every $e \in \Gen{T}$ of sort $E$ is $N$-free, and 
       \item 
        if $u \in \Gen{T}$ is of sort $G$
        and $\zs \in \Ind(u)$ then for some $t \in T$, \;
        $\zs \in \Ind(t)$.
    \end{enumerate}
\end{theorem}
\vshort{ \vs-pf-msg }
\long\def \vlindicator {
\begin{proof}
  By induction on operations used to construct terms from elements of
  $T$.

  The interesting cases are when $u$ is of the form %
  $u_1 u_2$ or $t ^ e$ where $t$, $u_1$, $u_2$ and $e$ are each normal
  form terms in $\Gen{T}$.

    In the first case, then, $u$ is a product
    \[
    t_1 \gop \dots \gop t_n 
    \]
    where each factor comes from $u_1$ or $u_2$.  Since each $t_i$ is
    of the form $ v, \ginv(v), v^e, \text{ or } \ginv(v^e)$, the
    normal form of this term results by canceling any factors (from
    different $u_i$) that are inverses of each other.  %
    No new $E$-subterms are created, so no new indicator vectors are
    created, and our assertion follows.

    The other case is when $u$ is $t ^ e$.  Note that since $e$ is in
    $\Gen{T}$ we know that $e$ is $N$-free.  It suffices to show that %
    $\Ind(t^e) = \Ind(t)$.  %
    Letting $t$ be in normal form, $t^e$ is
    \[
    (t_1)^e \gop \dots \gop (t_n)^e 
    \]
    Each $(t_i)^e$ is of the form
    \[
    v^e \quad (\einv(v))^e \quad (v^{e'})^e \quad (\ginv(v^{e'}))^e
    \]
    The first two terms are $N$-free.    The second kind of term reduces to %
    $v^{e \emult e'}$, and the indicator set for this term is
    precisely %
    $\Ind(e)$ since $e'$ is $N$-free.   The last term reduces to 
    $\ginv(v^{e' \emult e})$  and we can argue just as in the previous case.

    The cases for concatenation and encryption are immediate from the
    induction hypothesis, since they simply propagate indicator
    vectors.
\end{proof}
} 
\vlong{ \vlindicator}

\paragraph{An Adversary Limitation.}
We justify now our central technique, that the adversary cannot
generate messages with new indicators, using variables of sort $E$
that are \emph{non-originating before} node $n$.

\begin{definition}
  A basic value $a$ is \emph{non-originating before} $n$ in bundle
  $\bnd$ if, for all $n'\preceq_{\bnd} n$, $a$ does not originate at
  $n'$.  

  The \emph{indicator basis} $\operatorname{IB}_{\bnd}(n)$ of node
  $n$, where $n$ is a node of $\bnd$, is the set:
  $$\{
  a \mbox{ of sort }E\colon a \mbox{ is non-originating before }n
  \}.$$
  We assume $\operatorname{IB}_{\bnd}(n)$ is ordered in some
  conventional way.
\end{definition}

\begin{theorem} \label{thm:indicators:preserved}\label{level-0} 
  Let $W$ be an adversary web of $\bnd$, and let $n$ be a transmission
  node of $W$, and let $N$ be a sequence of elements drawn from 
  $\operatorname{IB}_{\bnd}(n)$.  If $v\in\Ind_N(\msg(n))$, then there is a
  regular transmission node $n'\prec_{\bnd}n$ in $\bnd$ such that
  $v\in\Ind_N(\msg(n'))$.
\end{theorem}
\vshort{ \vs-pf-msg }
\long\def \vlindicatorspreserved{ %
\begin{proof}
  Let $T_R$ be the set of messages received on $W$, and let %
  $T_M$ be the set of basic values emitted by $W$; set %
  $T = T_R \cup T_M$. %
  The message $u=\msg(n)$ is in $\Gen{T}$.  %
  The set $T_R$ is $N$-free, as a consequence of the fact that every
  message received on $W$ must have originated, and $T_M$ is $N$-free
  since it is a set of basic values not in $N$ (indeed, each term in
  $T_M$ has an empty indicator set). %
  So Theorem~\ref{indicator} applies.  Since each $t \in T_M$ has
  empty indicator set we conclude that every indicator in $u$ comes
  from a message in $T_R$, as desired.
\end{proof}
} 
\vlong{ \vlindicatorspreserved }

In {\iadh} protocols, every message from regular participants has
indicators in $\{\indZero\}, \{\indA\}, \{\indB\}, \{\indX\},
\{\indY\}$, etc.  Since the adversary can never transmit a message
with any indicators he has not received, no messages with other
indicators will ever be sent or received.  Messages encrypted using
keys derived from Diffie-Hellman values preserve this property.  Using
Thm.~\ref{thm:indicators:preserved} and
Assumption~\ref{assumption:regular:cert},
\ref{assumption:regular:cert:unique:P}:
\begin{corollary}\label{cor:ca} 
  Let $\bnd$ be a bundle for an {\iadh} protocol using certificates
  $\tagged{\mathsf{cert}\; g^a\cons P}{\CA}$ and
  $\tagged{\mathsf{cert}\; g^\alpha\cons P'}{\CA}$.  If $a\in\non_{\bnd}$
  and $\Ind_{\seq{a}}(\alpha)\not=\indZero$, then $\alpha=e$ and $P=P'$.
\end{corollary}
%


\section{Analyzing IADH Protocols} 
\label{sec:analysis}

We now embark on analyzing {\iadh} protocols, focusing on $\um$,
$\mqv$, and $\cremers$.  We aim to illustrate the way that our
algebraic tools---normal forms and indicators---work together with the
more familiar tools of symbolic protocol analysis.  These are notions
such as causal well-foundedness that are basic to strand spaces.  We
start with properties for which the indicators bear the main burden.
In Section~\ref{sec:implicit} we turn to implicit authentication.  It
requires subtler proofs, which are more sensitive to the details of
the key computation.  
%

We believe that our presentation of security goals is a contribution
in itself.  They appear to us to be clear distillations of the
structural elements in the goals, which have often appeared in more
cluttered forms---particularly obscured by more operational ideas---in
some of the literature.

We start though with a useful lemma about the session keys produced by
regular strands, saying that they always reflect the parameters of
that strand.
\begin{lemma}\label{lemma:session:keys} 
  Let protocol $\Pi$ be an {\iadh} protocol, but possibly without
  Assumption~\ref{assumption:regular:init:resp},
  Clauses~\ref{assumption:freshness:init}
  and~\ref{assumption:freshness:resp}.

  Suppose $\bnd$ is a $\Pi$-bundle, and $s$ is a $\Pi$ initiator or
  responder strand with long term secret $a$ and ephemeral value $x$,
  succeeding with key $K$:
  \begin{description}
    \item[$\Pi$ is $\um$:] If $x\in\non_{\bnd}$, then for
    $K=H({Y_B}^a\cons {R_B}^x)$, we have $\indX\in\Ind_{\seq{x}}(K)$.
    If $a\in\non_{\bnd}$, then $\indA\in\Ind_{\seq{a}}(K)$.
    \item[$\Pi$ is $\mqv$:] If $x\in\non_{\bnd}$, then for 
    $K=({R_B}\cdot{Y_B}^{\gtoe{{R_B}}})^{s_A}$, we have 
    $\indX\in\Ind_{\seq{x}}(K)$.  If $a\in\non_{\bnd}$, then
    $\indA\in\Ind_{\seq{a}}(K)$.
    \item[$\Pi$ is $\cremers$:] If $x\in\non_{\bnd}$, then for
    $K=({R_B}\cdot{Y_B})^{x+a}$, we have $\indX\in\Ind_{\seq{x}}(K)$.
    If $a\in\non_{\bnd}$, then $\indA\in\Ind_{\seq{a}}(K)$.
  \end{description}
%
\end{lemma}
\vshort{ \vs-pf-msg }
\long\def \lemmasessionkeysproof{ %
\begin{proof}
  For $\um$, $a$ or $x$ can cancel only if $s$ receives a value $R_B$
  or $Y_b$ with indicator $\seq{-1}$ for $a$ or $x$, resp.  Hence
  there is some earlier node $m$ on which some message with indicator
  $\seq{-1}$ was transmitted, and let $m_0$ be a minimal such node.

  However, by the definitions, $m_0$ is not a regular node, which
  transmit only values with non-negative indicators.  By
  Thm.~\ref{thm:indicators:preserved}, $m_0$ cannot be an adversary
  node either, when $a$ or $x\in\non_{\bnd}$ resp.  

  For $\mqv$, let $R_B=g^\eta$, where $\eta$ is a possibly compound
  value the adversary may have engineered, and let $Y_B=g^\beta$.  Now
  $K=g^{x\eta}\cdot g^{a\eta\gtoe{g^x}} \cdot g^{x\beta\gtoe{g^\eta}}
  \cdot g^{a\beta\gtoe{g^x}\gtoe{g^\eta}}$.  $K$ may be $a$-free
  because $g^{a\eta\gtoe{g^x}}$ and
  $g^{a\beta\gtoe{g^x}\gtoe{g^\eta}}$ cancel.  This occurs when
  ${a\eta\gtoe{g^x}}=-{a\beta\gtoe{g^x}\gtoe{g^\eta}}$,
  i.e.~$\eta=-\beta\gtoe{g^\eta}$.  However, in this case $x$ also
  cancels out, as ${x\eta}=-{x\beta\gtoe{g^\eta}}$.  So the exponent
  is 0 and $K=1$, contradicting the assumption that strand $s$
  delivers a successful key.  

  $\mqv$ could also cancel if $R_B$ or $Y_b$ has indicator $\seq{-1}$,
  but this is excluded by the same argument as with $\um$.

  The argument for $\cremers$ is the same as for $\mqv$.  
\end{proof}
} 
\vlong{\lemmasessionkeysproof}

\paragraph{Key Secrecy and Impersonation.}
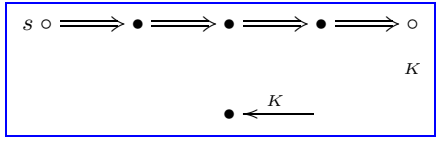
\begin{figure}[tb]
  \centering\frameit{\xymatrix{
      s \; \circ\ar@{=>}[r] & 
    \bullet\ar@{=>}[r] & \bullet\ar@{=>}[r] &
    \bullet\ar@{=>}[r] & \circ\ar@{}[d]|{K}\\
    & & \bullet & \null\ar[l]_{K} &
  }}  
  \caption{Key secrecy:  This diagram cannot occur}
  \label{fig:key:secrecy}
\end{figure}
In Fig.~\ref{fig:key:secrecy} we present the core idea of key secrecy.
Suppose that the upper strand $s$ is an initiator or responder run
that ends by computing session key $K$.  Moreover, suppose that a
listener strand is present, which receives $K$.  Then, if the long
term secrets $a,b\in\non$, this diagram cannot be completed to a
bundle $\bnd$.  This holds even without the freshness assumptions on
regular initiator and responder strands.
\begin{goal}[Key Secrecy]
  Suppose that $\bnd$ is a $\Pi$-bundle with $a,b\in\non_{\bnd}$, and
  strand $s$ is a $\Pi$ initiator or responder strand with long term
  secret parameter $a$ and long term peer public value $Y=g^b$.  Then
  $\bnd$ does not contain a listener $\bullet\leftarrow K$.
\end{goal}

\begin{theorem} \label{thm:key-secrecy}
  Let protocol $\Pi$ be an {\iadh} protocol using any of the key
  computation methods in Eqns.~\ref{eq:dh:unif:model:key:comp},
  \ref{eq:mqv:key:comp}, \ref{eq:dh:cremers:feltz:key:comp}, but
  possibly without Assumption~\ref{assumption:regular:init:resp},
  Clauses~\ref{assumption:freshness:init}
  and~\ref{assumption:freshness:resp}.  Then $\Pi$ achieves the
  security goal of key secrecy.  
\end{theorem}

\begin{proof}
  For sake of contradiction suppose that $\bullet\leftarrow K$ is in
  $\bnd$.  Then $K$ is transmitted on some node. %
  Computing indicators relative to the basis $\seq{a,b}$, $K$ has
  indicator $\seq{1,1}$ (by
  Eqns.~\ref{eq:dh:unif:model:key:comp:opt}--\ref{eq:dh:mqv:key:comp:opt},
  \ref{eq:dh:cremers:feltz:key:comp:opt} and
  Lemma~\ref{lemma:session:keys}). %
  By Thm.~\ref{thm:indicators:preserved}, some regular node transmits
  a message with indicator $\seq{1,1}$. %
  But this is a contradiction, since regular strands transmit only
  values with indicators $\seq{0,0}$ and, during certification,
  $\seq{1,0}$ and $\seq{0,1}$.
  %
\end{proof}
Curiously, resistance to impersonation attacks concerns the same
diagram, Fig.~\ref{fig:key:secrecy}, although with different
assumptions.  An impersonation attack is a case in which the
adversary, having compromised $A$'s long term secret $a$, uses it to
obtain a session key $K$, while causing $A$ to have a session yielding
$K$ as session key.  If $A$'s session uses $Y_B=g^b$, where $b$ is the
uncompromised long term secret of $B$, then the adversary has
succeeded in \emph{impersonating} $B$ to $A$.\footnote{By contrast, it
  is hopeless---when $a$ is compromised---to try to prevent the
  adversary from impersonating $A$ to others.}
The protocols $\mqv$ and $\cremers$ resist impersonation attacks, but
$\um$ does not.  In this result, we rely here on the freshness
assumptions on regular initiator and responder strands,
Assumption~\ref{assumption:regular:init:resp},
Clauses~\ref{assumption:freshness:init}
and~\ref{assumption:freshness:resp}.  We are in effect trading off an
assumption on a long term secret for assumptions on the ephemeral
values.
\begin{goal}[Resisting Impersonation] 
  Suppose that $\bnd$ is a $\Pi$-bundle with $b\in\non_{\bnd}$, and
  strand $s$ is a $\Pi$ initiator or responder strand using ephemeral
  secret $x$ and long term peer public value $Y=g^b$.  Then $\bnd$
  does not contain a listener $\bullet\leftarrow K$.
\end{goal}
\begin{theorem} \label{thm:resist} Let protocol $\Pi$ be an {\iadh}
  protocol using either of the two key computation methods in
  Eqns.~\ref{eq:mqv:key:comp} and \ref{eq:dh:cremers:feltz:key:comp}.
  Then $\Pi$ achieves the security goal of resisting impersonation.
\end{theorem}
\vshort{
\begin{proof}
Similar to the proof of Theorem~\ref{thm:key-secrecy}, using
  indicators relative to the basis $\seq{b,x}$.
\end{proof}
} 
\long\def \thmresistproof{ %
\begin{proof}
  For sake of contradiction suppose that $\bullet\leftarrow K$ is in
  $\bnd$.  Then $K$ is transmitted on some node. %
  When we compute indicators relative to the basis $\seq{b,x}$, $K$ has
  indicator $\seq{1,1}$ (by
  Eqns.~\ref{eq:dh:unif:model:key:comp:opt}--\ref{eq:dh:mqv:key:comp:opt},
  \ref{eq:dh:cremers:feltz:key:comp:opt} and
  Lemma~\ref{lemma:session:keys}). %
  By Thm.~\ref{thm:indicators:preserved} we conclude that some regular
  node transmits a message with indicator $\seq{1,1}$. %
  But this is a contradiction, since regular strands transmit only
  values with indicators $\seq{0,0}$ and, during certification,
  $\seq{1,0}$ and $\seq{0,1}$.
  %
  %
  %
\end{proof}
} 
\vlong{ \thmresistproof } 

This argument does not apply to $\um$, because its key
$K=H(g^{ab}\cons g^{xy})$ has indicators $\{\seq{1,0}, \seq{0,1}\}$ in
this basis.  Thus, Theorem~\ref{thm:indicators:preserved} buys us
nothing.  In fact, $\um$ fails to prevent impersonation attacks.  

\paragraph{Forward Secrecy.}  Forward secrecy is generally described
as preventing disclosure of the session key of a session, if the
long-term secrets of the regular participants in that session are
compromised subsequently.  We consider two different versions of the
forward secrecy property.  The first may be called \emph{weak} forward
secrecy, and all of our {\iadh} protocols achieve it.  We present weak
forward secrecy in Fig.~\ref{fig:weak:forward}.  Essentially, weak
forward secrecy holds because the non-originating ephemeral values
$x,y$ prevent the adversary from computing the session key.  Thus,
Assumption~\ref{assumption:regular:init:resp},
Clauses~\ref{assumption:freshness:init}
and~\ref{assumption:freshness:resp} are essential.
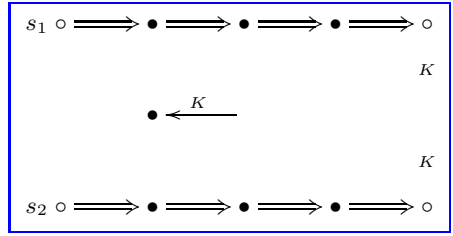
\begin{figure}[tb]
    $$\frameit{\xymatrix{
        s_1\; \circ\ar@{=>}[r] & \bullet\ar@{=>}[r] &
        \bullet\ar@{=>}[r] &
        \bullet\ar@{=>}[r] & \circ\ar@{}[d]|{K} \\
        & \bullet & \null\ar[l]_{K} & & \\
        s_2\;\circ\ar@{=>}[r] & \bullet\ar@{=>}[r] &
        \bullet\ar@{=>}[r] & \bullet\ar@{=>}[r] & \circ \ar@{}[u]|{K}
      }}$$
  \caption{Weak forward secrecy:  This diagram cannot occur}
  \label{fig:weak:forward}
\end{figure}
\begin{goal}[Weak Forward Secrecy]
  Suppose that $\bnd$ is a $\Pi$-bundle, and strands $s_1,s_2$ are
  distinct $\Pi$ initiator or responder strands, issuing the same
  session key $K$.  Then $\bnd$ does not contain a listener
  $\bullet\leftarrow K$.  
\end{goal}
\begin{theorem}
  Let protocol $\Pi$ be an {\iadh} protocol using any of the key
  computation methods in Eqns.~\ref{eq:dh:unif:model:key:comp},
  \ref{eq:mqv:key:comp}, \ref{eq:dh:cremers:feltz:key:comp}.  Then
  $\Pi$ achieves the weak forward secrecy security goal.  
\end{theorem}

\vshort{
\begin{proof}
  Just as for Theorems~\ref{thm:key-secrecy} and~\ref{thm:resist},
  using indicators with the basis $\seq{x,y}$; $K$ has indicator
  $\seq{1,1}$.
\end{proof}
}
\vlong{ %
\begin{proof} 
Just as for Theorems~\ref{thm:key-secrecy} and~\ref{thm:resist}: in this
  case compute
  indicators relative to the basis $\seq{x,y}$, and note that $K$ has
  indicator $\seq{1,1}$ yet
  regular strands transmit only
  values with indicators $\seq{0,0}$ and, during certification,
  $\seq{1,0}$ and $\seq{0,1}$.

\end{proof}
} 
A stronger notion of forward secrecy stresses the word
\emph{subsequently}.  A local session occurs, and the compromise of
the long term keys happens after that session is finished:  Can the
adversary then retrieve the session key?  We formalize this idea in a
diagram in which the long term secrets $a,b$ are transmitted after a
session issuing in session key $K$ completes.  
\begin{figure}[tb]
  $$\frameit{\xymatrix{
      \circ\ar@{=>}[r] & \bullet\ar@{=>}[r] & \bullet\ar@{=>}[r] &
      \bullet\ar@{=>}[r]
      & \circ \\
      & & && \null\ar@{--}[dllll]\\
      &\bullet&\null\ar[l]_{K}& & \bullet\ar[l]_{\sk(B),b,a\quad} }}$$
  \caption{Strong Forward Secrecy:  This diagram cannot occur}
  \label{fig:strong:forward}
\end{figure}
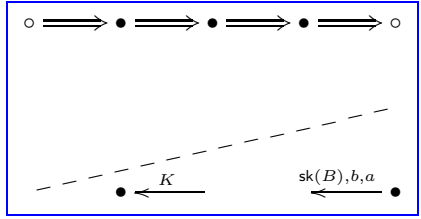
Moreover, we assume that the long term secrets are \emph{uniquely
  originating}.  This implies that they cannot have been used before
the session completed, which is exactly the intended force of
considering a \emph{subsequent compromise}.

Figure~\ref{fig:strong:forward} illustrates this situation.  The
slanted dotted line separates past from future, meaning that any event
northwest of the dotted line occurs before any event southwest of it.
This ordering relation between the end of the strand and the point of
disclosure is essential to the idea.  Also essential is
$a,b,\sk(B)\in\unique$, where $\sk(B)$ is $B$'s signing key.  $\mqv$
and $\um$ do not achieve perfect forward secrecy.  $\cremers$, like
the Station-to-Station protocol (Eqn~\ref{eq:dh:sts:auth}), does, for
a similar reason.
\begin{goal}[Forward Secrecy]
  Suppose that $\bnd$ is a $\Pi$-bundle with
  $a,b,\sk(B)\in\unique_{\bnd}$, and strand $s$ is a $\Pi$ initiator
  or responder strand using long term secret $a$ and long term peer
  public value $Y=g^b$.  Suppose that $\bullet\rightarrow a,b$ occurs
  subsequent to the last reception on $s$.  Then $\bnd$ does not
  contain a listener $\bullet\leftarrow K$.
\end{goal}
\begin{theorem}\label{thm:cf:forward:secrecy} 
  Let protocol $\Pi$ be the $\cremers$ protocol, with the ephemeral
  values $R_A,R_B$ signed as $\tagged{R_A}A$ and $\tagged{R_B}B$.
  Then $\Pi$ achieves the forward secrecy goal.
\end{theorem}
\vshort{ \vs-pf-msg }
\long\def\thmcfforwardsecrecy {%
\begin{proof}
  Since $\tagged{R_B}B$ is received on a node of $s$, and there is no
  compromise of $B$'s signing key until it has been received, there
  has been a regular node transmitting $\tagged{R_B}B$.  This follows
  from the Honest Ideal Theorem~\cite{ThayerHerzogGuttman99} or the
  Authentication Test Principle~\cite{GuttmanThayer02}.  

  Since a signed value $\tagged{R_B}B$ is transmitted only on a
  regular initiator or responder strand, we know that $R_B=g^y$ for
  some $y\in\non_{\bnd}$.  We may now take indicators relative to
  $\seq{x,y}$, and the rest of the proof proceeds as before.
\end{proof}
}
\vlong{ \thmcfforwardsecrecy } 
\vshort{The proof uses the signatures to establish $R_A$ and $R_B$ as
  the values $g^x$ and $g^y$ originating on the peer strands.  }
\vlong{Given the absence of signed units in $\mqv$ and $\um$, they have no
analog to the first step of this proof.  
}

\section{The Implicit Authentication Goal}  
\label{sec:implicit} 

Implicit authentication has been controversial, with a distinction
between ``implicit key authentication'' and ``resisting unknown
key-share
attacks''~\cite{blake1999authenticated,law2003efficient,kaliski2001unknown}.

The essential common idea is expressed in Figure~\ref{fig:impl:auth}.
\begin{figure}[tb]
  $$\frameit{\xymatrix{
      \circ\ar@{=>}[r] & \bullet\ar@{=>}[r]\ar@{}[d]|{[A,B',\ldots]} &
      \bullet\ar@{=>}[r] & \bullet\ar@{=>}[r] & \circ \ar@{}[d]|{K} 
      \\
      &&&& \\
      \circ\ar@{=>}[r] & \bullet\ar@{=>}[r]\ar@{}[u]|{[A',B,\ldots]} &
      \bullet\ar@{=>}[r] & \bullet\ar@{=>}[r] & \circ\ar@{}[u]|{K}
    }}$$
  \caption{Implicit authentication: In this diagram, $A=A'$ and
    $B=B'$}
  \label{fig:impl:auth}
\end{figure}
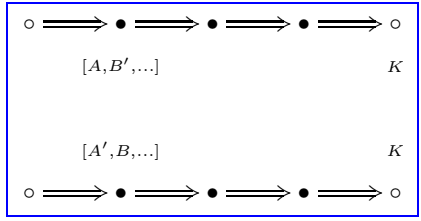
It shows two strands that compute the same session key $K$.  One has
parameters $[A,B',\ldots]$ and the other has parameters
$[A',B,\ldots]$, where we assume that the parameter for the
initiator's name appears first ($A,A'$) and parameter for the
responder's name appears second ($B',B$).  The authentication property
is that the participants agree on each other's identities, so that the
responder has the correct opinion about the initiator's identity and
\emph{vice versa}.

Implicit key authentication and resisting unknown key share attacks
differ in what non-compromise assumptions they make.

Resistance to unknown key-share attacks is the property that $A=A'$
and $B=B'$ whenever $a,b\in\non$.  The weaker assertion, implicit key
authentication, is that $A=A'$ and $B=B'$ whenever $a,b,a'\in\non$.
The additional non-compromise assumption is about $a'$, the long term
secret of the principal $E$ that $B$ \emph{thinks} he is communicating
with:  
\begin{quote}
  by definition the provision of implicit key authentication is
  considered only where $B$ engages in the protocol with an honest
  entity (which $E$ isn't).  \cite{blake1999authenticated}
\end{quote}
Law et al.~\cite{law2003efficient} use similar language.  Resisting
unknown key share attacks is simpler and more robust, and we will
refer to it as \emph{implicit authentication} (without ``key'').

\begin{goal}[Implicit Authentication] 
  Suppose that $\bnd$ is a $\Pi$-bundle with
  $a,b,\sk(B)\in\non_{\bnd}$, and strands $s_1,s_2$ are $\Pi$
  initiator and responder strands with parameters
  $[A,B',a,x,Y_{B'},R_{B'}]$ and $[A',B,b,y,Y_{A'},R_{A'}]$ resp.,
  where $s_1,s_2$ both yield session key $K$.  Then $A=A'$ and $B=B'$.
\end{goal}
\emph{Weak implicit authentication} states that $A=A'$, under the
extra assumption that there exists an $a'\in\non_{\bnd}$ such that
$Y_{A'}=g^{a'}$.  Symmetrically, $B=B'$, under the extra assumption
that there exists a $b'\in\non_{\bnd}$ such that $Y_{B'}=g^{b'}$.

We will prove four results.  We will show that $\um$ and $\cremers$
achieve implicit authentication.  Moreover, $\mqv$ achieves weak
implicit authentication.  Finally, (strong) implicit authentication
holds for $\mqv$, under an additional assumption.
%

\vlong{Of these protocols, $\um$ allows the simplest proof.}
\begin{theorem}\label{thm:um:impl:auth}
  $\um$ achieves implicit authentication.  
\end{theorem}
\vshort{ \vs-pf-msg }
\long\def\thmumimplauthproof{%
\begin{proof}
  Let $s_1,s_2$ be strands in $\bnd$ as in the implicit authentication
  goal, where also $a,b\in\non_{\bnd}$.  Since $s_1$ receives a
  certificate $\tagged{\mathsf{cert}\; Y_{B'}\cons B'}{\CA}$, by
  Assumption~\ref{assumption:regular:cert}, $\sk(\CA)\in\non_{\bnd}$.
  Hence, there was a certifying strand that transmitted this
  certificate, and by~\ref{assumption:regular:cert},
  Cl.~\ref{assumption:regular:cert:in:group}, $Y_{B'}=g^{b'}$ for some
  $b'$.
  By symmetry, $Y_{A'}=g^{a'}$.  

  The key computation, with the injectiveness of ${}\cons{}$ and
  $H$,\footnote{In the symbolic model, hash functions are modeled as
    injective.}  ensures $g^{a'b}=g^{ab'}$, hence $a'b=ab'$.  Thus,
  there is some $c$ such that $a'=ca$ and $b'=cb$.  Applying
  Cor.~\ref{cor:ca}, 
%
%
  we conclude $B'=B$.  Symmetrically, $A'=A$.
\end{proof}
} 
\vlong{\thmumimplauthproof}
{\smallskip\noindent{}Using indicators in a richer way than
  previously, we obtain:}
\begin{theorem}\label{thm:cremers:impl:auth}
  $\cremers$ achieves implicit authentication, when the strands
  $s_1,s_2$ receive $\tagged{R_{B'}}{B'}$ and $\tagged{R_{A'}}{A'}$,
  and $\sk({A'}),\sk({B'})\in\non_{\bnd}$.
\end{theorem}
%
\begin{proof}
  We start with $a,b\in\non_{\bnd}$, and
  Assumption~\ref{assumption:regular:init:resp} tells us
  $x,y\in\non_{\bnd}$.  Using the signatures, there exist regular
  initiator or responder strands transmitting $\tagged{R_{B'}}{B'}$
  and $\tagged{R_{A'}}{A'}$.  Hence, by
  Assumption~\ref{assumption:regular:init:resp}, ${R_{A'}}=g^{x'}$ and
  ${R_{B'}}=g^{y'}$, where $x',y'\in\non_{\bnd}$ and
  $g^{x'},g^{y'}\in\unique_{\bnd}$.  We may also use the certificates
  \vlong{(as in the previous proof)} to infer that $Y_{B'}=g^{b'}$ and
  $Y_{A'}=g^{a'}$.  Also $g^{a'},g^{b'}\in\unique_{\bnd}$.  Since the
  strands compute the same session key,
  \begin{equation}
    \label{eq:cf:proof} 
    g^{xy'}\cdot g^{xb'}\cdot g^{y'a}\cdot g^{ab'} 
    = g^{x'y}\cdot g^{x'b}\cdot g^{ya'}\cdot g^{a'b} 
  \end{equation}
  \vlong{We also know that none of these parameters can be replaced by
    a compound expression, since they are independently chosen on
    regular strands.  Moreover, none of $x,y,x',y'$ can equal any of
    $a,b,a',b'$, as $g^x,g^y,g^{x'},g^{y'}$ are uniquely originating,
    on initiator or responder strands.  The exponentials of the latter
    all originate on certificate request strands.}

  \vlong{Moreover, i}\vshort{\noindent I}f $x=y$, then $s_1=s_2$, so
  $A=A'$ and $B=B'$, and authentication is assured.  So assume
  $x\not=y$.

  We compute indicators for the four monomials on each side of
  Eqn.~\ref{eq:cf:proof}, as shown in Table~\ref{tab:cf:indicators}.
  We use as basis the non-originating parameters $a,b,x,y,x',y'$, in
  this order.  Since we do not know whether the primed variables equal
  their unprimed counterparts, there are undetermined entries (?) in
  the indicator vectors; an integer 0 or 1 shows the definite presence
  or absence of a parameter.
  \begin{table}[tb]
    \centering
    \begin{tabular}{cccc}
      ${xy'}$ & ${xb'}$ & ${y'a}$ & ${ab'}$ \\ 
      $\seq{0,0,1,?,?,1}$ & 
      $\seq{?,?,1,0,?,0}$ & 
      $\seq{1,0,0,?,?,1}$ &
      $\seq{1,?,0,0,0,0}$ \\[2mm] 
      $\seq{0,0,?,1,1,?}$ & 
      $\seq{0,1,?,0,1,?}$ & 
      $\seq{?,?,0,1,0,?}$ & 
      $\seq{?,1,0,0,0,0}$ \\
      ${x'y}$ & ${x'b}$ & ${ya'}$ & ${a'b}$ 
    \end{tabular}\\[2mm]
    \caption{Indicator vectors for $\cremers$ authentication}
    \label{tab:cf:indicators}
  \end{table}

  In the table, every vertically aligned pair is compatible, i.e.~we
  can fill in the undetermined entries so as to make the vectors
  agree.  Moreover, if two vectors are not vertically aligned, they
  are incompatible.  For instance, the rightmost entries have 0s for
  all the slots for ephemeral parameters, which put them in conflict
  with all of the other vectors.  

  Hence, ${xy'}={x'y}$,\ldots, ${ab'}={a'b}$.  Since each of these is
  a parameter, and not compound, we have $x=x',y=y',a=a',b=b'$.
  Applying Cor.~\ref{cor:ca}, $A=A'$ and $B=B'$.
\end{proof}
%
\vshort{A similar, though more elaborate, argument
  establishes:}\vlong{Turning now to {\mqv}:} 
%

\begin{theorem}\label{thm:mqv:weak} 
    $\mqv$ achieves {weak} implicit authentication.  
\end{theorem}
\vshort{ \vs-pf-msg }
\long\def\thmmqvweakproof{%
\begin{proof}
  Let $s_1,s_2$ be strands in $\bnd$ as in the weak implicit
  authentication goal, where also $a,b,a'\in\non_{\bnd}$ and
  $Y_{A'}=g^{a'}$.  Here the starting point is weaker than in
  $\cremers$, since we do not know that $R_{B'},R_{A'}$ originate on
  regular strands; we know only that they are group elements, so of
  the form $g^\eta,g^{\psi}$, resp., for $\psi,\eta:\NZE$.  Likewise,
  $Y_{B'}$, having been certified, is some group element $g^\beta$.
  Since $s_1,s_2$ yield the same key, we have:
  \begin{eqnarray*}
    && g^{x\eta}\cdot g^{a\eta\gtoe{g^x}} \cdot g^{x\beta\gtoe{g^\eta}}
    \cdot g^{a\beta\gtoe{g^x}\gtoe{g^\eta}}   \\ 
    &=& g^{\psi y}\cdot g^{\psi b\gtoe{g^y}} \cdot
    g^{a'y\gtoe{g^\psi}} \cdot g^{a'b\gtoe{g^y}\gtoe{g^\psi}}
  \end{eqnarray*}
  An adversary strategy for solving this consists of an assignment of
  possibly compound expressions to the Greek letters
  $\psi,\eta,\beta$.  The adversary wins if both sides of this
  equation reduce to the same normal form, but without forcing $A=A'$.

  We write the indicator vectors for this in
  Table~\ref{tab:mqv:w:indicators}, relative to the basis
  $\seq{a,a',b,x,y}$, all in $\non$.
  \begin{table}[tb]
    \centering
    \begin{tabular}{cccc}
      ${x\eta}$ & ${a\eta\gtoe{g^x}}$ & ${x\beta\gtoe{g^\eta}}$ & ${a\beta\gtoe{g^x}\gtoe{g^\eta}}$ \\ 
      $\seq{?,?,?,1,?}$ & 
      $\seq{1,?,?,?,?}$ & 
      $\seq{?,?,?,1,?}$ &
      $\seq{1,?,?,?,?}$ \\[2mm] 
      $\seq{?,?,?,?,1}$ & 
      $\seq{?,?,1,?,?}$ & 
      $\seq{?,1,?,0,1}$ & 
      $\seq{?,1,1,0,0}$ \\
      ${\psi y}$ & ${\psi b\gtoe{g^y}}$ & ${a'y\gtoe{g^\psi}}$ &
      ${a'b\gtoe{g^y}\gtoe{g^\psi}}$  
    \end{tabular}\\[2mm]
    \caption{Indicator vectors for $\mqv$ weak authentication}
    \label{tab:mqv:w:indicators}
  \end{table}
  There are many entries $?$, because we do not know whether $a=a'$,
  or what the adversary incorporated into the Greek letters
  $\beta,\eta,\psi$.  Nevertheless, the lower right entry has 0 for
  the $x$ slot, so it cannot equal the first or third entry in the
  first row, in which the $x$ slot is 1.  This leaves two
  possibilities, the second and fourth terms.

  In these terms, the $a$ slot is 1.  Thus, either $a'=a$ or $b=a$.
  If $a'=a$, we may apply Cor~\ref{cor:ca}.  

  So assume $a'\not=a$ and $b=a$.  If we choose term~2,
  i.e.~${a'b\gtoe{g^y}\gtoe{g^\psi}}={a\eta\gtoe{g^x}}$, then
  $\eta=a'r$, where $r$ is the ratio of boxed terms.  Turning to the
  term $x\eta$, we have $x\eta=xa'r$, i.e.~its $y$ and $b$ slots are
  0.  Thus, it cannot equal any of the monomials on the RHS.

  Choosing term~4,
  ${a'b\gtoe{g^y}\gtoe{g^\psi}}={a\beta\gtoe{g^x}\gtoe{g^\eta}}$, then
  $\beta=a' r$, where $r$ is a ratio of boxed values.  But since
  $\beta$ was certified, we can apply Cor~\ref{cor:ca} to infer that
  $\beta=a'$.  Plugging in, we now have ${a'y\gtoe{g^\psi}}$ with
  indicator $\seq{0,1,0,0,1}$.  Since ${x\beta\gtoe{g^\eta}}$ has
  indicator $\seq{0,1,0,1,0}$, there is no term in the top row that
  ${a'y\gtoe{g^\psi}}$ can match.  
\end{proof}
}
\vlong{ \thmmqvweakproof }
Kaliski~\cite{kaliski2001unknown} showed implicit authentication does
not hold for {\mqv}.  An adversary, observing $A$'s ephemeral public
value $R_A=g^x$, may generate a new $R_E$ depending on $R_A$ and
$Y_A=g^a$, and then a new long-term $Y_E$:
\begin{equation}
  \label{eq:kaliski}
  R_E = g^x\cdot(g^a)^{\gtoe{g^x}}\cdot g^{-1} \qquad\quad
  Y_E=g^{\gtoe{R_E}^{-1}}.
\end{equation}
Thus, $R_E=g^{x+a{\gtoe{g^x}}-1}$.  The adversary asks {\CA} to
certify $Y_E$, successfully proving possession of ${\gtoe{R_E}^{-1}}$.
This is compatible with our assumptions as
$\Ind_{\seq{a}}(Y_E)=\indZero$.  

$E$'s operations cancel out, so the certificate misleads $B$ into
thinking $K$ is shared with $E$, when it is shared with $A$.  A
mischievous priest $E$ can cause a criminal $B$ to believe $K$ shared
with $E$, when it fact it is shared with the district attorney $A$.
$E$ can thus induce $B$ to misdeliver a confession to $A$, leading to
an unexpected plot twist in Hitchcock's movie with Montgomery
Clift~\cite{Hitchcock53}.
\begin{definition}  
  Strands $s,d$ with parameters $[\ldots,a,x,\ldots]$ and
  $[\ldots,Y_A,R_A]$ are a \emph{doping pair} if $x$ appears in $Y_A$.

  Bundle $\bnd$ \emph{respects ephemerals} if no doping pair in
  $\bnd$ yields a shared key $K$.
\end{definition}
Doping, which~\cite{kaliski2001unknown} uses, is not visible to the
principal executing $d$.  We mention below 
%
a way to prevent it.
\begin{theorem}\label{thm:mqv:impl:auth}
  Suppose $\bnd$ is an $\mqv$ bundle that respects ephemerals.  Then
  $\bnd$ satisfies (full) implicit authentication.
\end{theorem}
\vshort{ \vs-pf-msg }
\long\def\thmmqvimplauth{%
\begin{proof}
  Let $s_1,s_2$ be strands in $\bnd$ as in the implicit authentication
  goal, where $a,b\in\non_{\bnd}$ and $Y_{A'}=g^{\alpha},
  Y_{B'}=g^{\beta}$, for $\alpha,\beta:\NZE$.  $R_{B'},R_{A'}$ are
  group elements of the form $g^\eta,g^{\psi}$, resp., for
  $\psi,\eta:\NZE$.  Since $s_1,s_2$ yield the same key,
  \begin{eqnarray}
    \label{eq:mqv:auth}
    &&   g^{x\eta}\cdot g^{a\eta\gtoe{g^x}} \cdot g^{x\beta\gtoe{g^\eta}}
    \cdot g^{a\beta\gtoe{g^x}\gtoe{g^\eta}}  \nonumber \\
    &=& g^{\psi y}\cdot g^{\psi b\gtoe{g^y}} \cdot
    g^{\alpha y\gtoe{g^\psi}} \cdot g^{\alpha b\gtoe{g^y}\gtoe{g^\psi}}
  \end{eqnarray}
  By Cor.~\ref{cor:ca}, either $\alpha=a$ or $\alpha=b$ or
  $\Ind_{\seq{a,b}}(\alpha)=\indZero$.  Likewise, either $\beta=a$ or
  $\beta=b$ or $\Ind_{\seq{a,b}}(\beta)=\indZero$.

  If both $\alpha,\beta\in\{a,b\}$, then we have a case of weak
  authentication from both sides, so Thm.~\ref{thm:mqv:weak} gives the
  desired result.  Assume then that at least one, e.g.~$\alpha$, has
  $\Ind_{\seq{a,b}}(\alpha)=\indZero$.  Since $\bnd$ respects
  ephemerals, $(s_1,s_2)$ is not a doping pair, and $(s_2,s_2)$ is not
  a doping pair, so $x,y$ are syntactically absent from $\alpha$.  So
  in fact $\Ind_{\seq{a,b,x,y}}(\alpha)=\indZero$.
  By Equation~\ref{eq:mqv:auth}
  \begin{eqnarray*}
    && x\eta + a\eta\gtoe{g^x} + x\beta\gtoe{g^\eta} +
    a\beta\gtoe{g^x}\gtoe{g^\eta}  \nonumber \\
    & = & \psi y + \psi b\gtoe{g^y} + \alpha y\gtoe{g^\psi} + 
    \alpha b\gtoe{g^y}\gtoe{g^\psi} 
  \end{eqnarray*}
  The Greek letters may be compound expressions.  Thus, $\alpha
  y\gtoe{g^\psi}$ and $\alpha b\gtoe{g^y}\gtoe{g^\psi}$ may each yield
  a number of monomials when reduced to normal forms.  However,
  because $\alpha$ has indicator $\indZero$, and the boxed terms have
  indicator $\indZero$, monomials resulting from $\alpha
  y\gtoe{g^\psi}$ all have indicator $\seq{0,0,0,1}$.
  Monomials resulting from $\alpha b\gtoe{g^y}\gtoe{g^\psi}$ all have
  indicator $\seq{0,1,0,0}$.

  When the LHS normalizes, no monomial on the LHS can have indicator %
  $\seq{0,1,0,0}$ or  $\seq{0,0,0,1}$ %
  because each one has a factor of $x$ or $a$.
  So, the last two summands on the RHS cannot contribute any monomials
  to the normal form.

  By Lemma~\ref{lemma:session:keys}, the LHS has non-zero
  contributions of $a$ and $x$.  Hence, $\psi$ must have non-zero
  contributions of them.  

%

  We write $\psi$ as the sum $\psi=\psi_{nz}+\psi_0$, where
  $\psi_{nz}$ collects all the monomials in $\psi$ with non-zero
  indicators, and $\psi_0$ collects all those with indicator
  $\indZero$.

  In particular, $\psi_0 y$ must cancel $\alpha y\gtoe{g^\psi}$, and
  $\psi_0 b\gtoe{g^y}$ must cancel $\alpha b\gtoe{g^y}\gtoe{g^\psi}$.
  Each of these leads to the conclusion:
  \begin{align} \label{kaliski-eqn}
    -(\psi_0/\gtoe{g^\psi})=\alpha.
  \end{align}
  
  Hence the normal form of $\psi_0$ must be some $\phi\gtoe{g^\psi}$, so
  that $\gtoe{g^\psi}$, which has occurrences of $x$, can syntactically
  cancel.  Hence, the normal form of $\psi$ is
  $$\psi_{nz}+\phi\gtoe{g^\psi},$$
  contradicting the well-foundedness of the syntactic terms.  
\end{proof}
} 
\vlong{ \thmmqvimplauth } 
\vlong{The preceding analysis sheds some light on Kaliski's
  attack~(\ref{eq:kaliski}) on MQV.  There, equation~\ref{kaliski-eqn}
  holds with $\psi_0=-1$ and $\alpha=\gtoe{g^\psi}^{-1}$.  However, we
  here have the additional assumption above that $\bnd$ respects
  ephemerals: since $s_1,s_2$ is not a doping pair, $\alpha$ can have
  no occurrence of $x$, but as we have observed, $\psi$ must.}

The interesting approaches to preventing the Kaliski attack---that is,
to ensure that executions respect ephemerals---involve time and
causality.  Suppose that the $\CA$ always takes at least a minimum
time $t_C$ between receiving a certification request and issuing the
certificate.  Moreover, the initiator always times out and discards a
session if it does not complete within a period $t_I$, where
$t_I<t_C$.  For instance, if $T_C$ is an hour and $t_I$ is a half
hour, this approach would be practically workable.  No synchronization
between different principals is required for this, since each
participant makes purely local decisions about timing.  Non-malicious
sessions would be entirely unaffected.  Then, in any completed
session, no certified value can involve an ephemeral in that session,
since it cannot yet have been generated at the time the value was
certified.

\section{Uniform Equality and the Completeness of \dhthy}
\label{sec:algebra}
\label{sec:axiomatization}
In this section we justify the use of \dhthy, specifically the use of
\dhthy-normal forms to model messages.  Any theorem of \dhthy surely
holds in all DH-structures.  Theorem~\ref{tfae} gives us the
converse, namely that every equation that holds in all DH-structures
is a theorem of \dhthy.  Indeed, given a non-principal ultraflter $D$
over the set of primes, there is a single structure
{\ultraD} that is ``generic'' for all of the DH-structures:  An
equation $s=t$ is valid in {\ultraD} if and only if it is valid in
infinitely many DH-structures.

We work first with models of the language of \dhthy but with the
\boxfn removed from the signature.  They have all the structure
required to analyze $\um$ and $\cremers$.  We then extend our results
to DH-structures equipped with a \boxfn function.

Algebraically isomorphic structures can have very different
\emph{computational} properties.  Indeed, the prime field $\fieldq$
presented as the group of integers mod $q$ induces a DH-structure
where the base group is the \emph{additive} group of $\fieldq$ and
exponentiation is multiplication.  The discrete log problem in this
structure is computationally tractable.  However, $\fieldq$ is
isomorphic to a subgroup of order $q$ of the \emph{multiplicative}
group of integers modulo some prime $p$.  There, the discrete log
problem may be intractable.  Although the algebra is blind to the
computational distinctions, we focus here on the algebraic equations
between terms in DH-structures.

First, we show that the field of scalars, i.e.~the exponents, carries
all the {algebraic} information in a model of \dhthy.

\begin{definition}\label{model-from-field}
  Let $F$ be a field.  We construct a \boxfn-free model $\mdlF$ of
  theory \dhthy as follows.
  The sorts $E$ and $G$ are each interpreted as the domain of $F$; the
  sort $NZE$ is interpreted as the set of non-0 elements of $E$.  The
  operations of $E$ are interpreted just as in $F$ itself.  The
  operation $\gop$ is taken to be $\eadd$ from $E$, thus $\gid$ and
  $\ginv$ are taken to be {0 and $\eneg$}.  Exponentiation is
  multiplication:  $a^e$ is interpreted as $a \emult e$.
\end{definition}
For each field $F$, any $\mdlF$ satisfies all of the equations in
$\dhthy$.  When $F$ is the prime field of order $q$ then $\mdlF =
\mdl_{\fieldq}$ is, up to isomorphism, precisely the standard DH
algebra of order $q$. %
When $F$ is the additive group of rational numbers then $\mdlF =
\mdlQ$ will be of interest to us below.

The key device for reasoning about uniform equality across DH-structures
is the notion of \emph{ultraproduct}, cf.~e.g.~\cite{chang_model-theory}.  We
let the variable $D$ range over non-principal ultrafilters over the set
of prime numbers.

\begin{definition} \label{def:std-models} Let $D$ be a non-principal
  ultrafilter over the set of prime numbers and let %
  $\fieldprod$ be the ultraproduct structure %
  \mbox{$\prod_D \dset{\fieldq}{q \text{ prime}}.$}
  $\mdl_{\fieldprod}$ is the DH structure obtained from $\mathbb{F}_D$
  via Definition~\ref{model-from-field}.  For simplicity we write
  $\ultraD$ for $\mdl_{\fieldprod}$.
\end{definition}

The crucial facts about ultraproducts for our purposes are: (i) a
first-order sentence is true in an ultraproduct if and only if the set
of indices at which it is true is a set in $D$; (ii) when $D$ is
non-principal, every cofinite set is in $D$.  We show below that the
set of \emph{equations} valid in $\ultraD$ does {not} depend on which
non-principal $D$ we use.


\long\def\algebralemmas{
$\fieldprod$ is a field, since each $\fieldq$ satisfies the first-order
axioms for fields.  \fieldprod has characteristic 0, since each equation 1
+ 1 + \dots + 1 = 0 is false in all but finitely many $\fieldq$.
Indeed, it is false in all but one $\fieldq$.

\begin{lemma} \label{embedQ}
  The structure %
  \mdlQ can be embedded as a submodel in any  $\ultraD$.
\end{lemma}
\begin{proof}
  Since $\fieldprod$ has characteristic 0, and $\rat$ is the prime
  field of characteristic 0, $\rat$ is embeddable in $\fieldprod$.  The
  models $\ultraD$ and $\rat$ are definitional expansions of
  \fieldprod and \rat, so the embedding of \rat into $\fieldprod$ extends
  to embed $\mdlQ$ into $\ultraD$.
\end{proof}

\begin{lemma} \label{rat-distinct} Let $t:G$ be in normal form, in the
  \boxfn-free sublanguage of \dhthy.    There exists an environment 
  $\eta : \vars \to \rat$ such that if $u$ and $u'$ are distinct
  subterms of $t$, $\eta (u) \neq \eta(u')$ in $\mdlQ[D]$.
\end{lemma}

\begin{proof}
 In the structure \mdlQ, exponentiation is interpreted as
  multiplication, so it suffices to consider the expression obtained by
  replacing %
  $\gop$ and $\ginv$ by $+$ and $-$, and the exponentiation operator by
  $\emult$, and viewing $t$ as an ordinary rational expression in
  several variables $x_1, \dots, x_k$ (the variables occurring in $t$).
  We may view $t$ as determining a real function $ f_t : \Rea^{k} \to
  \Rea$.  In fact each subterm $t'$ of $t$ similarly determines a
  function from $\Rea^{k}$ to $\Rea$ (not all variables of $t$ will
  occur in all subterms, but we may still treat each as inducing a
  $k$-ary function).   
  So the family of subterms of $t$ determines a (finite) set of rational
  functions, and we can find a rational point $\vec{r} = (r_1, ...,
  r_k)$ such that no two of these functions agree on $\vec{r}$.
We define $\eta$ to map each $x_i$ to $r_i$.
\end{proof}
}%

\vlong{ \algebralemmas }

\vshort{Using the model $\mdlQ$ over the rationals $\mathbb{Q}$, we
  obtain:}

\begin{corollary} \label{rat-complete}
  If $s$ and $t$ are distinct normal forms then it is not the case that
  $\mdlQ \models s=t$.
\end{corollary}
\long\def \vlratdistinct {\begin{proof}
  Form the term $u \equiv s \gop \ginv(t)$.  Since $s$ and $t$ are
  distinct normal forms this term is in normal form.  By
  Lemma~\ref{rat-distinct} there is an environment $\env$ with $\env(s)
  \neq \env(t)$, and the result follows.
\end{proof}}
\vlong{\vlratdistinct}

\dhthy is complete for uniform equality in the absence of the
\boxfn-function:
\begin{theorem} \label{tfae}
  For each pair of $G$-terms $s$ and $t$ in the \boxfn-free fragment of \dhthy, 
  the following are equivalent
  \begin{enumerate}
  \item  \label{pvble} %
    $\dhthy \vdash s = t$
  \item \label{idbox} %
   For all  $q$, $\mdl_{\fieldq} \models s = t$
   \item \label{ultra:all} %
   For all non-principal $D$, $\ultraD\models s = t$
   \item \label{ultra:some} %
   For some non-principal $D$, $\ultraD\models s = t$
%
  \item \label{ratbox} %
    $\mdlQ \models s=t$
  \item \label{nf} %
    if $s$ reduces to $s'$ with $s'$ irreducible, and
    $t$ reduces to $t'$ with $t'$ irreducible, then $s'$ and $t'$ are
    identical modulo associativity and commutativity of $\gop$, $\eadd$,
    and $\emult$.
  \end{enumerate}
\end{theorem}
\begin{proof}
  It suffices to establish the cycle of entailments %
  \ref{pvble} implies 
  \ref{idbox} \dots{} implies \ref{nf} implies \ref{pvble}.  %
  The first three of these steps are immediate, %
  as is the fact that \ref{nf} implies \ref{pvble}.  %
  The fact that \ref{ultra:some} implies \ref{ratbox} follows
  \vlong{from Lemma~\ref{embedQ}}\vshort{by considering $\mdlQ$}.  To
  conclude \ref{nf} from \ref{ratbox}, use
  Corollary~\ref{rat-complete}.
%
\end{proof}
As a corollary of \refthm{tfae}, these equivalences hold for $E$-term
equations as well.  Given terms $e$ and $e'$, form the equation $g^e =
g^{e'}$.  It is provable iff $e=e'$ is provable, and is true in a
given model \mdl\ iff $e=e'$ is.

\begin{corollary}\label{cor:inf:many:means:all} 
  If $\mdl_{\fieldq} \models s = t$ holds for infinitely many $q$,
  then for all $q$, $\mdl_{\fieldq} \models s = t$.  
\end{corollary}
\begin{proof}
  Suppose that $\{q\colon \mdl_{\fieldq} \models s = t\}$ is infinite.
  Then there is a non-principal ultrafilter $D$ containing this set.
  So (\ref{ultra:some}) in Thm.~\ref{tfae} holds, and we apply
  (\ref{ultra:some})$\Rightarrow$(\ref{idbox}).
\end{proof}

The equivalence of $\dhthy$-provability with equality in the models 
is the technical core of our claim that $\dhthy$ captures ``uniform
equality.''

The model $\mdlQ$ is convenient: this single model, based on a
familiar structure, serves to witness uniform equality simplifies
analyses.  Our first analysis of {\mqv} used this.


The model \ultraD satisfies an even more striking property.  It follows
from results of Ax~\cite{ax_elementary-theory} that the first-order
theory of \ultraD is decidable.  So the structure $\ultraD$ is an
attractive one for closer study of the ``uniform'' properties of
DH-structures.

\paragraph{Incorporating \boxfn.}
An analogue of Theorem~\ref{tfae} holds for the full language of
\dhthy, the language appropriate for reasoning about \mqv.
\vlong{The starting point is like Lemma~\ref{rat-distinct}}.

\long\def \vlelimbox{
\begin{lemma} \label{elim-box}
 Let $t:G$ be in normal form, in the language of \dhthy.    There exists 
  an interpretation of the \boxfn function and an environment \env such
  such that
  if $u$ and $u'$ are distinct subterms of $t$, $\eta (u) \neq \eta(u')$ in $\mdlQ$.
\end{lemma}

\begin{proof}
  The proof is by induction on the number of \boxfn-subterms of $t$.  If
  this number is 0 then we may apply Lemma~\ref{rat-distinct} and simply
  use the following simple \boxfn function: 
$ \gtoe{a} = {a}$  if $a \neq 0$ and $\gtoe{0} =1$.

  Otherwise let $\gtoe{s}$ be a subterm of $t$ such that $s$ is
  \boxfn-free.  Let $t'$ be the term obtained from $t$ by replacing each
  occurrence of $\gtoe{s}$ by a variable $v$ occurring nowhere in $t$. %
  Then $t'$ is in normal form, so by induction there is a function
  $\boxfn_{0}$  and an
  environment \env that acts as an injection over the subterms of $t'$.
  We may assume that \env is defined on all the variables of $t$ (even
  though some may not occur in $t'$).    We claim that we can define
  \boxfn so that the resulting function, taken with the same
  environment \env satisfies the Lemma.    We define $\boxfn$ to agree
  with $\boxfn_0$ on all values except $\env(s)$, where we put
  $\gtoe{\env(s)} = \env(v)$.    Since $\env$ is guaranteed to yield
  different values on distinct subterms of $t'$, the use of $\boxfn_0$
  will yield the same values as the use of $\boxfn$ on subterms of $t$
  other than $\gtoe{s}$.
\end{proof}

} 
\vlong{ \vlelimbox }
By an argument similar to that establishing Corollary~\ref{rat-complete}
we obtain
\begin{corollary} \label{rat-box-complete}
  $\dhthy \proves s=t$ iff for all \boxfn functions
  $\mdlQ \models s=t$.
\end{corollary}
\noindent{}From this follows, finally: 

\begin{theorem} \label{full-tfae}
  For each pair of $G$-terms $s$ and $t$ in the full language of \dhthy
  the following are equivalent
  \begin{enumerate}
  \item  \label{pvble:full} %
    $\dhthy \vdash s = t$
  \item \label{idbox:full} %
   For all  $q$ and all \boxfn functions on $\mdl_{\fieldq}$, $\mdl_{\fieldq} \models s = t$
   \item \label{ultra:all:full} %
   For all non-principal $D$, for all \boxfn functions on $\ultraD$,  $\ultraD\models s = t$
   \item \label{ultra:some:full} %
   For some non-principal $D$, and all \boxfn functions on
   $\ultraD$, $\ultraD\models s = t$
%
  \item \label{ratbox:full} %
    for all \boxfn functions on \mdlQ,    $\mdlQ \models s=t$
  \item \label{nf:full} %
    if $s$ reduces to $s'$ with $s'$ irreducible, and
    $t$ reduces to $t'$ with $t'$ irreducible, then $s'$ and $t'$ are
    identical modulo associativity and commutativity of $\gop$, $\eadd$,
    and $\emult$.
  \end{enumerate}
%
\end{theorem}
\vshort{ \vs-pf-msg }
\long\def \vlfulltfae{
\begin{proof}
  As for Theorem~\ref{tfae} we can establish a cycle of entailments.
  The non-trivial changes to the arguments presented for
  Theorem~\ref{tfae} are
\begin{itemize}
  \item to conclude \ref{ratbox:full} from \ref{ultra:some:full} now,
  we observe that given a \boxfn-function on \mdlQ that entails $\mdlQ
  \models s \neq t$ we can, via the embedding of $\mdlQ$ into
  $\ultraD$, construct a \boxfn-function on \ultraD such that
  $\ultraD \models s \neq t$, and
  \item to conclude \ref{nf:full} from \ref{ratbox:full} now, we use
  Corollary~\ref{rat-box-complete}.
\end{itemize}
\end{proof}
}%
\vlong{ \vlfulltfae }


\section{Conclusion and Related Work} 
\label{sec:conc}

\paragraph{Related Work.}  Within the symbolic model, there has been
substantial work on some aspects of DH, starting with Boreale and
Buscemi~\cite{boreale2003symbolic}, which provides a symbolic
semantics~\cite{amadio2003symbolic,FioreAbadi01,MillenShmatikov01} for
a process calculus with algebraic operations for DH.  Their symbolic
semantics is based on unification.

Indeed, symbolic approaches to protocol analysis have relied on
unification as a central part of their reasoning.  Goubault-Larrecq,
Roger, and Verma~\cite{goubault2005abstraction} use a method based on
Horn clauses and resolution modulo AC, providing automated proofs of
passive security.  Maude-NPA~\cite{escobar2008state,escobar2009maude}
is also usable to analyze many protocols involving DH, again depending
heavily on unification.  

All of these approaches appear to face a fundamental problem with a
theory like the {\dhthy} theory of Section~\ref{sec:algebra}, in which
it would be unwise to rely on the decidability of the unifiability
problem.  Unifiability is undecidable in the theory of rings%
, essentially by the unsolvability of Hilbert's tenth problem.  There
are, however, many related theories for which undecidability is not
known, for instance the diophantine theory of the
rationals~\cite{bergstra_rational-numbers}.

K\"usters and Truderung~\cite{kuesters2009using} finesse this issue by
rewriting protocol analysis problems.  The original problems use an AC
theory involving exponentiation.  They transform it into a
corresponding problem that does not require the AC property, and so
can work using standard ProVerif resolution~\cite{Blanchet01}.  Their
approach covers a surprising range of protocols, although,
like~\cite{chevalier2003deciding}, not {\iadh} protocols such as
$\mqv$ or $\cremers$.

Another contrast between this paper and previous work is the uniform
treatment of numerous security goals.  Our methods are applicable to
confidentiality, authentication, and further properties such as
forward secrecy.  

Our adversary model is active.  For passive attacks, there has been
some work on computational soundness for Diffie-Hellman, with Bresson
et al.~\cite{BressonEtAl11} giving an excellent treatment.


\paragraph{Conclusion and Future Work.}  %
In this paper, we have applied the strand space framework to {\iadh}
protocols, such as {\um}, {\cremers}, and {\mqv}, establishing about a
dozen security properties of them.  While all of them have been
previously claimed, few have been proved in as informative a way as we
do here.  Moreover, our proofs rely on a few fundamental principles
that can be easily applied.  They combine rewriting techniques and the
\emph{indicator} idea.

We also provided a deeper model-theoretic treatment that justifies our
rewriting theory with respect to an adversary model.  Our adversary
can use any algebraic facts that are true in all but finitely many
DH-structures.  Since other cryptographic primitives such as bilinear
pairings are built by enriching DH-structures, it is highly desirable
to have proof techniques that work in this rich algebraic framework.

Connecting this with the standard computational model remains for
future work.  In our model the adversary must choose its whole
strategy before seeing the concrete messages for a particular run, or
even knowing the prime $q$.  This raises the question of the
\emph{computational soundness} of our approach, a focus of future
research: Does the Decisional Diffie-Hellman assumption ensure that
the adversary gets no asymptotic advantage from knowing $q$ and the
concrete messages?

Our proofs here are handcrafted.  However, we are currently pursuing
an approach using model-finding in \emph{geometric logic}, a
generalization of Horn logic, which offers great promise for
mechanizing many of these conclusions.


\paragraph{Acknowledgments.}  
We gratefully acknowledge support by the National Science Foundation
under grant CNS-0952287.  We are grateful to Shriram Krishnamurthi,
Moses Liskov, Cathy Meadows, John Ramsdell, Paul Rowe, Paul Timmel,
and Ed Zieglar for extremely helpful and often vigorous discussions.

{ 
  \bibliographystyle{plain}
  \bibliography{../bibtex/secureprotocols,../bibtex/dd}

\begin{thebibliography}{10}

\bibitem{amadio2003symbolic}
R.M. Amadio, D.~Lugiez, and V.~Vanack{\`e}re.
\newblock On the symbolic reduction of processes with cryptographic functions.
\newblock {\em Theoretical Computer Science}, 290(1):695--740, 2003.

\bibitem{ankney1995unified}
R.~Ankney, D.~Johnson, and M.~Matyas.
\newblock The unified model. contribution to ansi x9f1.
\newblock {\em Standards Projects (Financial Crypto Tools), ANSI X}, 42, 1995.

\bibitem{ax_elementary-theory}
James Ax.
\newblock The elementary theory of finite fields.
\newblock {\em The Annals of Mathematics}, 88(2):pp. 239--271, 1968.

\bibitem{bergstra_rational-numbers}
Jan~A. Bergstra and J.~V. Tucker.
\newblock The rational numbers as an abstract data type.
\newblock {\em Journal of The ACM}, 54, 2007.

\bibitem{blake1999authenticated}
Simon Blake-Wilson and Alfred Menezes.
\newblock {Authenticated Diffe-Hellman key agreement protocols}.
\newblock In {\em Selected Areas in Cryptography}, pages 630--630. Springer,
  1999.

\bibitem{Blanchet01}
Bruno Blanchet.
\newblock An efficient protocol verifier based on {Prolog} rules.
\newblock In {\em 14th Computer Security Foundations Workshop}, pages 82--96.
  IEEE CS Press, June 2001.

\bibitem{boreale2003symbolic}
M.~Boreale and M.G. Buscemi.
\newblock {Symbolic analysis of crypto-protocols based on modular
  exponentiation}.
\newblock {\em Mathematical Foundations of Computer Science 2003}, pages
  269--278, 2003.

\bibitem{BressonEtAl11}
Emmanuel Bresson, Yassine Lakhnech, Laurent Mazar{\'e}, and Bogdan Warinschi.
\newblock Computational soundness: The case of {Diffie-Hellman} keys.
\newblock In Veronique Cortier and Steve Kremer, editors, {\em Formal Models
  and Techniques for Analyzing Security Protocols}, Cryptology and Information
  Security Series. {IOS} Press, 2011.

\bibitem{chang_model-theory}
C.C. Chang and H.J. Keisler.
\newblock {Model Theory, volume 73 of Studies in Logic and the Foundations of
  Mathematics}, 1990.

\bibitem{chevalier2003deciding}
Yannick Chevalier, Ralf K{\"u}sters, Micha{\"e}l Rusinowitch, and Mathieu
  Turuani.
\newblock {Deciding the security of protocols with Diffie-Hellman
  exponentiation and products in exponents}.
\newblock {\em FST TCS 2003: Foundations of Software Technology and Theoretical
  Computer Science}, pages 124--135, 2003.

\bibitem{CremersFeltz2011}
Cas Cremers and Michele Feltz.
\newblock One-round strongly secure key exchange with perfect forward secrecy
  and deniability.
\newblock Cryptology ePrint Archive, Report 2011/300, 2011.
\newblock \url{http://eprint.iacr.org/2011/300}.

\bibitem{Cremers06}
C.J.F. Cremers.
\newblock {\em Scyther - Semantics and Verification of Security Protocols}.
\newblock {Ph.D.} dissertation, Eindhoven University of Technology, 2006.

\bibitem{DiffieHellman76}
W.~Diffie and M.~Hellman.
\newblock New directions in cryptography.
\newblock {\em IEEE Transactions on Information Theory}, 22(6):644--654,
  November 1976.

\bibitem{diffie1992authentication}
Whitfield Diffie, Paul~C. {van Oorschot}, and Michael~J. Wiener.
\newblock {Authentication and authenticated key exchanges}.
\newblock {\em Designs, Codes and Cryptography}, 2(2):107--125, 1992.

\bibitem{duran_church-rosser-checker}
Francisco Dur{\'a}n and Jos{\'e} Meseguer.
\newblock A {Church-Rosser} checker tool for conditional order-sorted
  equational {Maude} specifications.
\newblock In Peter~Csaba {\"O}lveczky, editor, {\em WRLA}, volume 6381 of {\em
  Lecture Notes in Computer Science}, pages 69--85. Springer, 2010.
\newblock Version 3j, available at \url{http://maude.lcc.uma.es/CRChC}.

\bibitem{escobar2008state}
Santiago Escobar, Catherine Meadows, and Jos\'e Meseguer.
\newblock {State space reduction in the Maude-NRL protocol analyzer}.
\newblock {\em Computer Security-ESORICS 2008}, pages 548--562, 2008.

\bibitem{escobar2009maude}
Santiago Escobar, Catherine Meadows, and Jos\'e Meseguer.
\newblock {Maude-{NPA}: Cryptographic protocol analysis modulo equational
  properties}.
\newblock {\em Foundations of Security Analysis and Design V}, pages 1--50,
  2009.

\bibitem{FioreAbadi01}
Marcelo Fiore and Mart{\'\i}n Abadi.
\newblock Computing symbolic models for verifying cryptographic protocols.
\newblock In {\em Computer Security Foundations Workshop}, June 2001.

\bibitem{aprove}
J.~Giesl, P.~Schneider-Kamp, and R.~Thiemann.
\newblock Aprove 1.2: Automatic termination proofs in the dependency pair
  framework.
\newblock In {\em Proceedings IJCAR '06}, LNAI 4130, pages 281--286. Springer,
  2006.

\bibitem{goubault2005abstraction}
Jean Goubault-Larrecq, Muriel Roger, and Kumar Verma.
\newblock {Abstraction and resolution modulo AC: How to verify
  Diffie-Hellman-like protocols automatically}.
\newblock {\em Journal of Logic and Algebraic Programming}, 64(2):219--251,
  2005.

\bibitem{Guttman10}
Joshua~D. Guttman.
\newblock Shapes: Surveying crypto protocol runs.
\newblock In Veronique Cortier and Steve Kremer, editors, {\em Formal Models
  and Techniques for Analyzing Security Protocols}, Cryptology and Information
  Security Series. {IOS} Press, 2011.

\bibitem{Guttman11a}
Joshua~D. Guttman.
\newblock State and progress in strand spaces: Proving fair exchange.
\newblock {\em Journal of Automated Reasoning}, 2011.
\newblock Accepted, March 2010. DOI: 10.1007/s10817-010-9202-1.

\bibitem{GuttmanThayer02}
Joshua~D. Guttman and F.~Javier {{Thayer}}.
\newblock Authentication tests and the structure of bundles.
\newblock {\em Theoretical Computer Science}, 283(2):333--380, June 2002.
\newblock Conference version appeared in \emph{IEEE Symposium on Security and
  Privacy}, May 2000.

\bibitem{Hitchcock53}
Alfred Hitchcock.
\newblock I confess.
\newblock Warner Brothers, March 1953.
\newblock \url{http://www.imdb.com/title/tt0045897/}.

\bibitem{kaliski2001unknown}
Burton~S. Kaliski.
\newblock {An unknown key-share attack on the MQV key agreement protocol}.
\newblock {\em {ACM} Transactions on Information and System Security},
  4(3):275--288, 2001.

\bibitem{kapur2003unification}
Deepak Kapur, Paliath Narendran, and Lida Wang.
\newblock {An E-unification algorithm for analyzing protocols that use modular
  exponentiation}.
\newblock {\em RewritingTechniques and Applications}, pages 150--150, 2003.

\bibitem{krawczyk2005hmqv}
H.~Krawczyk.
\newblock {HMQV: A high-performance secure Diffie-Hellman protocol}.
\newblock In {\em Advances in Cryptology--CRYPTO 2005}, pages 546--566.
  Springer, 2005.

\bibitem{kunz2006security}
Sebastian Kunz-Jacques and David Pointcheval.
\newblock {About the Security of MTI/C0 and MQV}.
\newblock {\em Security and Cryptography for Networks}, pages 156--172, 2006.

\bibitem{kuesters2009using}
Ralf K{\"u}sters and Tomasz Truderung.
\newblock Using {ProVerif} to analyze protocols with {Diffie-Hellman}
  exponentiation.
\newblock In {\em {IEEE} Computer Security Foundations Symposium}, pages
  157--171. IEEE, 2009.

\bibitem{law2003efficient}
L.~Law, A.~Menezes, M.~Qu, J.~Solinas, and S.~Vanstone.
\newblock {An efficient protocol for authenticated key agreement}.
\newblock {\em Designs, Codes and Cryptography}, 28(2):119--134, 2003.

\bibitem{menezes2007another}
Alfred Menezes.
\newblock Another look at {HMQV}.
\newblock {\em Journal of Mathematical Cryptology}, 1:47--64, 2007.

\bibitem{MillenShmatikov01}
Jonathan~K. Millen and Vitaly Shmatikov.
\newblock Constraint solving for bounded-process cryptographic protocol
  analysis.
\newblock In {\em 8th ACM Conference on Computer and Communications Security
  (CCS '01)}, pages 166--175. ACM, 2001.

\bibitem{cpsa09}
John~D. Ramsdell and Joshua~D. Guttman.
\newblock {CPSA}: A cryptographic protocol shapes analyzer.
\newblock In {\em Hackage}. The MITRE Corporation, 2009.
\newblock \url{http://hackage.haskell.org/package/cpsa}; see esp.~\texttt{doc}
  subdirectory.

\bibitem{rubio_fully-syntactic}
Albert Rubio.
\newblock A fully syntactic {AC-RPO}.
\newblock In Paliath Narendran and Micha{\"e}l Rusinowitch, editors, {\em RTA},
  volume 1631 of {\em Lecture Notes in Computer Science}, pages 133--147.
  Springer, 1999.

\bibitem{ThayerHerzogGuttman99}
F.~Javier Thayer, Jonathan~C. Herzog, and Joshua~D. Guttman.
\newblock Strand spaces: Proving security protocols correct.
\newblock {\em Journal of Computer Security}, 7(2/3):191--230, 1999.

\end{thebibliography}
}

\vshort{\appendix
Here we collect the proofs that were either omitted or merely sketched in the body of the paper.

\subsection{An Equational Theory of Messages}
\noindent{\bf \refthm{thm:sn+cr}.   }
  The reduction $\to_{\dhthy}$ is terminating and confluent modulo AC.
  \vlsn+cr


\subsection{Indicators}

\bigskip
\noindent  {\bf \refthm{indicator}.  } %
  Suppose $T$ is a collection of terms such that every $e \in T$ of sort
  $E$ is $N$-free.  Then
    \begin{enumerate}
       \item 
        every $e \in \Gen{T}$ of sort $E$ is $N$-free, and 
       \item 
        if $u \in \Gen{T}$ is of sort $G$
        and $\zs \in \Ind(u)$ then for some $t \in T$, \;
        $\zs \in \Ind(t)$.
    \end{enumerate}
\vlindicator

\bigskip
\noindent{\bf \refthm{thm:indicators:preserved}.  }
 Let $W$ be an adversary web of $\bnd$, and let $n$ be a transmission
  node of $W$, and let $N$ be a sequence of elements drawn from 
  $\operatorname{IB}_{\bnd}(n)$.  If $v\in\Ind_N(\msg(n))$, then there is a
  regular transmission node $n'\prec_{\bnd}n$ in $\bnd$ such that
  $v\in\Ind_N(\msg(n'))$.
 \vlindicatorspreserved 

\subsection{Analyzing IADH Protocols}


\noindent{\bf \reflem{lemma:session:keys}.  }  Let protocol $\Pi$ be
an {\iadh} protocol, but possibly without
Assumption~\ref{assumption:regular:init:resp},
Clauses~\ref{assumption:freshness:init}
and~\ref{assumption:freshness:resp}.

  Suppose $\bnd$ is a $\Pi$-bundle, and $s$ is a $\Pi$ initiator or
  responder strand with long term secret $a$ and ephemeral value $x$,
  succeeding with key $K$:
  \begin{description}
    \item[$\Pi$ is $\um$:] If $x\in\non_{\bnd}$, then for
    $K=H({Y_B}^a\cons {R_B}^x)$, we have $\indX\in\Ind_{\seq{x}}(K)$.
    If $a\in\non_{\bnd}$, then $\indA\in\Ind_{\seq{a}}(K)$.
    \item[$\Pi$ is $\mqv$:] If $x\in\non_{\bnd}$, then for 
    $K=({R_B}\cdot{Y_B}^{\gtoe{{R_B}}})^{s_A}$, we have 
    $\indX\in\Ind_{\seq{x}}(K)$.  If $a\in\non_{\bnd}$, then
    $\indA\in\Ind_{\seq{a}}(K)$.
    \item[$\Pi$ is $\cremers$:] If $x\in\non_{\bnd}$, then for
    $K=({R_B}\cdot{Y_B})^{x+a}$, we have $\indX\in\Ind_{\seq{x}}(K)$.
    If $a\in\non_{\bnd}$, then $\indA\in\Ind_{\seq{a}}(K)$.
  \end{description}

\lemmasessionkeysproof
\medskip 

\noindent{\bf \refthm{thm:resist}.  }  Let protocol $\Pi$ be an
{\iadh} protocol using either of the two key computation methods in
Eqns.~\ref{eq:mqv:key:comp} and \ref{eq:dh:cremers:feltz:key:comp}.
Then $\Pi$ achieves the security goal of resisting impersonation.

\thmresistproof

\medskip 
\noindent{\bf \refthm{thm:cf:forward:secrecy}.  }  Let protocol $\Pi$
be the $\cremers$ protocol, with the ephemeral values $R_A,R_B$ signed
as $\tagged{R_A}A$ and $\tagged{R_B}B$.  Then $\Pi$ achieves the
forward secrecy goal.

\thmcfforwardsecrecy
\subsection{The Implicit Authentication Goal}  

\medskip 
\noindent{\bf \refthm{thm:um:impl:auth}.  }  $\um$ achieves implicit
authentication.

\thmumimplauthproof

\medskip 
\noindent{\bf \refthm{thm:mqv:weak}.  }  $\mqv$ achieves {weak}
implicit authentication.

\thmmqvweakproof

\medskip 
\noindent{\bf \refthm{thm:mqv:impl:auth}.  }  Suppose $\bnd$ is an
$\mqv$ bundle that respects ephemerals.  Then $\bnd$ satisfies (full)
implicit authentication.

\thmmqvimplauth

\subsection{Uniform Equality and the Completeness of \dhthy}

\algebralemmas

\vlelimbox

\medskip
\noindent{\bf \refcor{rat-complete}.  } %
  If $s$ and $t$ are distinct normal forms then it is not the case that
  $\mdlQ \models s=t$.

\vlratdistinct


\medskip
\noindent{\bf \refthm{full-tfae}.  } %
  For each pair of $G$-terms $s$ and $t$ in the full language of \dhthy
  the following are equivalent
  \begin{enumerate}
  \item  
    $\dhthy \vdash s = t$
  \item 
   for all  $q$ and all \boxfn functions on $\mdl_{\fieldq}$, $\mdl_{\fieldq} \models s = t$
   \item 
   For all non-principal $D$, for all \boxfn functions on $\ultraD$,  $\ultraD\models s = t$
   \item 
   For some non-principal $D$, and all \boxfn functions on
   $\ultraD$, $\ultraD\models s = t$
  \item 
    for all \boxfn functions on \mdlQ,    $\mdlQ \models s=t$
  \item 
    if $s$ reduces to $s'$ with $s'$ irreducible, and
    $t$ reduces to $t'$ with $t'$ irreducible, then $s'$ and $t'$ are
    identical modulo associativity and commutativity of $\gop$, $\eadd$,
    and $\emult$.
  \end{enumerate}
\vlfulltfae

}

\end{document}